\documentclass[authoryear,preprint,review,14pt]{elsarticle}

\usepackage{graphicx}
\usepackage{epsfig}
\usepackage{amsmath,amsfonts, amsthm, mathrsfs, amssymb, bbm}
\usepackage{natbib}
\usepackage{times}
\usepackage{subfigure}

\usepackage{multirow}
\usepackage{hyperref}

\usepackage{lipsum}
\usepackage{rotating}

\newtheorem{theorem}{Theorem}
\newtheorem{lemma}{Lemma}
\journal{Journal of Theoretical Biology}

\begin{document}

\begin{frontmatter}

\title{Mathematically Modeling Spillover Dynamics of Emerging Zoonoses with Intermediate Hosts}

\author[DartMath]{Katherine P. Royce}
\ead{Katherine.P.Royce.19@dartmouth.edu}
\author[DartMath,DartMed]{Feng Fu\corref{ff}}
\ead{feng.fu@dartmouth.edu}

\address[DartMath]{Department of Mathematics, Dartmouth College, Hanover, NH 03755, USA}

\address[DartMed]{Department of Biomedical Data Science, Geisel School of Medicine at Dartmouth, Lebanon, NH 03756, USA}

\cortext[ff]{Corresponding author at: 6188 Kemeny Hall, 27 N. Main Street, Hanover, NH 03755, USA. Fax:  +1 (603) 646 1312}

\begin{abstract}
The World Health Organization describes zoonotic diseases as a major pandemic threat, and modeling the behavior of such diseases is a key component of their control. Many emerging zoonoses, such as SARS, Nipah, and Hendra, mutated from their wild type while circulating in an intermediate host population, usually a domestic species, to become more transmissible among humans, and moreover, this transmission route will only become more likely as agriculture and trade intensifies around the world. Passage through an intermediate host enables many otherwise rare diseases to become better adapted to humans, and so understanding this process with mathematical epidemiological models is necessary to prevent epidemics of emerging zoonoses, guide policy interventions in public health, and predict the behavior of an epidemic. In this paper, we account for spillovers of a zoonotic disease mutating in an intermediate host by means of modeling transmission dynamics within and between three host species, namely, wild reservoir, intermediate domestic animals, and humans. We calculate the basic reproductive number of the pathogen, present critical conditions for the emergence dynamics of zoonosis, and perform stability analysis of admissible disease equilibria. Our analytical results agree well with long-term simulations of the system. We find that in the presence of biologically realistic interspecies transmission parameters, a zoonotic disease can establish itself in humans even if it fails to persist in its reservoir and intermediate host species. Our model and results can be used to understand the dynamic behavior of any zoonosis with intermediate hosts and assist efforts to protect public health.
\end{abstract}

\begin{keyword}
Zoonosis \sep Evolutionary Epidemiology \sep Pathogen Adaptation \sep Global Health \sep Mathematical Biology% 5 keywords
\end{keyword}

\end{frontmatter}

%% \linenumbers

%% main text

\section{Introduction} \label{ch:introduction}

Zoonotic diseases, which originate in animals and infect humans, are one of the most concerning epidemic threats of the 21$^{\textrm{st}}$ century and form 60\% of all known infectious diseases \citep{karesh2012ecology}. These pathogens cause a billion cases of illness per year, inflict severe economic damage, and pose an increasing threat in a more connected world; indeed, endemic zoonoses are currently the greatest global burden on human health \citep{karesh2012ecology}. Public health threats such as HIV-AIDS, avian influenza, SARS, Ebola, Nipah, Hendra, and rabies all trace their origin to nonhuman reservoir species, and it is likely that the next global pandemic will be a zoonosis \citep{karesh2012ecology}. Zoonoses have comprised a growing area of public health research for the last two decades \citep{daszak2000emerging}, and the World Health Organization even cites ``Disease X'', a pathogen currently unknown to cause human disease that might evolve to become more transmissible among humans, as a priority for research and development in pandemic prevention \citep{who2018}.

The frequency of new pathogens emerging into the human population$-$rapidly increasing in incidence or geographic range to become a threat to public health$-$is increasing \citep{morse2012prediction}, and zoonoses comprise 75\% of emerging infectious diseases \citep{woolhouse2005host}. Emergence of zoonoses is linked to human behavioral changes and increasing rates of interaction with wildlife, human travel, and global trade \citep{cunningham2017one}, as well as accelerating climate change (\citep{cunningham2017one}, \citep{lloyd2015nine}, \citep{wolfe2005bushmeat}). The dynamics of a zoonosis in its reservoir host are frequently cited as an influence on its emergence in humans \citep{karesh2012ecology}, but to our knowledge, no attempt has been made to quantify the entire course of an emerging zoonosis, from its origins in a wild reservoir host to an epidemic in humans. Indeed, \cite{lloyd2009epidemic} blames a desire to view zoonoses in a piecewise manner, as a concatenation of different epidemics rather than a connected system, for the lack of quantitative understanding of zoonoses as a new type of disease. Some of the most pressing unaddressed questions in establishing the mathematical theory of zoonoses include better capturing disease dynamics within nonhuman species in order to characterize changes in the disease before it infects humans; focusing on the first cases of human infection to understand how a pathogen actively adapts to humans; and developing a theory for the role of intermediate hosts in the emergence of the disease \citep{lloyd2015nine}. Despite these recognized challenges and the frequent use of mathematical biology to assist with risk assessment and surveillance strategies for other types of diseases, there is no unifying mathematical theory or set of principles that can be used to frame discussions of zoonotic spillovers \citep{lloyd2015nine}. This gap in modeling spillover dynamics limits our understanding of zoonoses, as does a general lack of mathematical modeling of multihost pathogens and quantification of the rate of human-to-human transmission (\citep{lloyd2009epidemic}, \citep{allen2012mathematical}). This paper provides such a mathematical model for a zoonosis emerging through an intermediate host.

Zoonotic diseases are currently classified on the basis of their human-to-human transmissibility \citep{lloyd2009epidemic}, which is assumed to be a critical distinction between pathogens with pandemic potential and pathogens that remain relatively rare (\citep{karesh2012ecology}, \citep{woolhouse2005host}, \citep{lloyd2015nine}). The major distinction in zoonotic spread within humans is whether the pathogen can spread beyond its primary individual host to infect other humans: whether the basic reproduction number $R_0$, the number of secondary cases produced by an index case in an entirely naive population, is greater than 1 \citep{lloyd2015nine}. This classification rests on a three-stage framework summarized by \cite{morse2012prediction}, \cite{lloyd2009epidemic}, and \cite{wolfe2005bushmeat}. Stage 1, pre-emergence, represents zoonoses circulating in an intermediate host but only capable of spillover into a dead-end human host, with no further transmission. Stage 2, localized emergence, defines diseases that can maintain stuttering chains in a human population with reinfection from animal hosts but are incapable of sustaining themselves in humans alone. Stage 3, pandemic emergence, classifies diseases that are fully adapted to humans and thus capable of causing outbreaks in our species alone (\citep{morse2012prediction}, \citep{lloyd2009epidemic}). In this paper, we examine the process of pathogen evolution through these different stages to show that with a mutation to a human-transmissible strain in an intermediate host, a pathogen can maintain an endemic equilibrium in humans even in stage 1, suggesting that the epidemiological stratification of zoonotic diseases based on their perceived threat to humans may be myopic.

\subsection{The Role of Intermediate Hosts for a Zoonosis}

In contrast to pathogens which evolved to infect humans, such as smallpox, the biology of emerging zoonoses is adapted to a different host species, called the reservoir host. Zoonoses are the product of a pathogen exploiting a new niche, sometimes one exposed by anthropogenic changes or induced by the amplification of its transmission \citep{karesh2012ecology}. Zoonotic pandemics occur when the pathogen gains the ability to circulate in a human population, rather than infrequently causing infection in an individual dead-end host \citep{richard2014avian}, a change which usually requires one or more mutations from the wild type \citep{lloyd2015nine}. While the change to a pathogen's $R_0$ in humans can take place over a single individual infection, this modification is considered to be a result of the role that different animal hosts play in amplifying or transmitting a zoonosis to humans \citep{karesh2012ecology}. Since a pathogen's transmissibility can also be affected by anthropogenic factors such as the host species' population structure or resource and habitat availability \citep{allen2012mathematical},
intermediate hosts$-$a non-reservoir animal species in which a zoonotic pathogen circulates$-$particularly domestic animals, provide greater opportunity for a pathogen to mutate to a human-transmissible form, because these species are biologically similar to the pathogen's wild reservoir and have greater contact with humans. It is therefore extremely important to develop a theory for a human-transmissible disease arising from a zoonotic pathogen in an intermediate host population; with such a framework, policymakers can move towards prevention of a human pandemic rather than amelioration of one \citep{lloyd2015nine}. 

As an example of the role of intermediate hosts, the adaptation of avian influenza, one of the most well-studied zoonoses, to humans requires a mutation in domestic pigs or poultry. Avian influenza's success in a new host species is governed by its receptor binding specificity \citep{richard2014avian}; with circulation in domestic pigs, which express  both human- and avian-influenza type receptors in their tracheae, the virus has an opportunity to mutate to a form that can infect humans (\citep{neumann2009emergence}, \citep{ma2009role}). Further, as domestic animals, swine have more contact with humans than wild birds do and can thus spread a disease more quickly \citep{ma2009role}. Domestic poultry can play a similar role for the disease, since circulation in a domestic poultry population may increase the pathogenicity of avian influenza among birds (\citep{vandegrift2010ecology}, \citep{ito2001generation}). As a result, human movement of livestock, not avian migration, is the dominant factor in the spread of highly pathogenic avian influenza, even though wild birds are the reservoir of the disease \citep{gauthier2007recent}. The influenzas are perhaps the easiest example to understand, as reassortment of different hemagglutinin and neuraminidase subtypes within one infected pig can produce entirely new pathogens \citep{neumann2009emergence}, but less drastic mutations can alter the transmissibility or lethality of any zoonosis. Pigs are an intermediate host for Nipah virus \citep{sharma2019emerging}, and the intensification of the pig industry in Malaysia was identified as the key factor in the spillover of the disease to humans in the 1990s \citep{cunningham2017one}. In this case, the disease dynamics that resulted from repeated introductions from bats, the pathogen's reservoir host, to pigs enabled Nipah to persist in its intermediate host and thus infect humans \citep{pulliam2011agricultural}. These examples support the general principle that the domestication of animals is linked to an increased risk of emergence of zoonotic diseases into the human population \citep{karesh2012ecology}, and Table~\ref{tab:zoonoses}, a sampling of zoonoses for which an intermediate host has been identified, shows the prevalence of domesticated species as intermediate hosts.

\begin{sidewaystable}[ht]
\centering
    \caption{Zoonotic diseases with intermediate hosts}
      \begin{small}
 \begin{tabular}{| l | c | c | l |}
    \hline
    Disease & Reservoir Host & Intermediate Host & Source\\
    \hline
    Nipah virus encephalitis & bats & pigs & \cite{karesh2012ecology}, \cite{cunningham2017one}, \cite{lloyd2015nine}, \cite{daszak2000emerging}\\
    Hendra virus disease & bats & horses & \cite{cunningham2017one}, \cite{lloyd2015nine}, \cite{daszak2000emerging}\\
    SARS & bats & civets & \cite{lloyd2015nine}\\
    Avian influenza & wild birds & domestic poultry, pigs &  \cite{vandegrift2010ecology}, \cite{ito2001generation}, \cite{iwami2007avian}\\
    Menangle virus disease & bats & pigs & \cite{cunningham2017one}, \cite{daszak2000emerging}\\
    Middle East Respiratory Syndrome & bats & camels & \cite{de2013mers} \\
    Campylobacteriosis & wild birds & domestic poultry & \cite{goodwin2012interdisciplinary}\\
    Japanese encephalitis & wild birds & pigs & 
    \cite{goodwin2012interdisciplinary}\\
    \hline
    \end{tabular}
   \end{small}
        \label{tab:zoonoses}
\end{sidewaystable}

%
%\begin{table}[!h]
%    \centering
%    \begin{small}
%    \begin{tabular}{| c | c | c | c|}
%    \hline
%    Disease & Reservoir Host & Intermediate Host & Source\\
%    \hline
%    Nipah virus encephalitis & bats & pigs & \cite{karesh2012ecology}, \cite{cunningham2017one}, \cite{lloyd2015nine}, \cite{daszak2000emerging}\\
%    Hendra virus disease & bats & horses & \cite{cunningham2017one}, \cite{lloyd2015nine}, \cite{daszak2000emerging}\\
%    SARS & bats & civets & \cite{lloyd2015nine}\\
%    Avian influenza & wild birds & domestic poultry, pigs &  \cite{vandegrift2010ecology}, \cite{ito2001generation}, \cite{iwami2007avian}\\
%    Menangle virus disease & bats & pigs & \cite{cunningham2017one}, \cite{daszak2000emerging}\\
%    Middle East Respiratory Syndrome & bats & camels & \cite{de2013mers} \\
%    Campylobacteriosis & wild birds & domestic poultry & \cite{goodwin2012interdisciplinary}\\
%    Japanese encephalitis & wild birds & pigs & 
%    \cite{goodwin2012interdisciplinary}\\
%    \hline
%    \end{tabular}
%        \end{small}
%    \caption{Zoonotic diseases with intermediate hosts}
%    \label{tab:zoonoses}
%
%\end{table}

\subsection{The Role of Mathematical Modeling}

Ordinary differential equations describing population dynamics and zoonotic transmission are a crucial tool in understanding the nonlinear interactions that are a hallmark of zoonotic diseases, a type of subgroup dynamics which can lead to counterintuitive behaviors (\citep{lloyd2009epidemic}, \citep{allen2012mathematical}). Mathematical models can enable experiments that would be unfeasible with real populations, predict future trends based on current data, and estimate key epidemic qualities such as the basic reproduction number of a pathogen in a specific population \citep{lloyd2009epidemic}. Explicitly quantifying the dynamics behind this adaptive transformation is thus critical to public health efforts; however, no previous models examine the changes an emerging zoonosis undergoes as it spreads between different species \citep{lloyd2015nine}.

There have been attempts to quantify the effect of pathogen mutations in humans alone. Models for tuberculosis sometimes include a distinction between  latent and active forms of the disease \citep{gumel2012causes}. \cite{iwami2007avian} recognizes that the ability of avian influenza to mutate during an epidemic is a crucial determinant of its pandemic potential, but conceptualizes this mutation as occuring within humans rather than another species, ignoring the intrinsically zoonotic behavior of the disease. \cite{gumel2009global} expands on this analysis by including a compartment for wild birds, but still locates the mutation after the pathogen's spillover to humans. This framework occludes the key population in the spread of a zoonosis: \cite{richard2014avian} cites two barriers, jumping to humans and efficient human-to-human transmission, that a zoonotic pathogen must overcome, and this change frequently occurs in the ``mixing vessel'' of an intermediate host species \citep{neumann2009emergence}. Further, controlling a human epidemic of a zoonotic disease depends on controlling the basic reproduction number in both animals and humans \citep{kim2010avian}, interventions not previously studied together. 

In this paper, we present a model which incorporates a pathogen mutation to a human-transmissible form in an intermediate host species, filling the gap noted by \cite{lloyd2015nine} with the introduction of a mathematical model that simulates the entire course of an emerging zoonosis. We investigate whether the presence of pathogen adaptation in intermediate hosts can amplify an epidemic among humans, with the goal of informing public health efforts to curb emerging infectious diseases.

The model presented here is based on the basic SIR model first presented by \cite{kermack1927contribution}, as well as the introduction to multihost SIR models presented by \cite{allen2012mathematical}. As a baseline and example, we use parameters that most closely reflect highly pathenogenic avian influenza, a classical example of a zoonosis with an intermediate host \citep{daszak2000emerging} and one for which the most data is available. Further, although we model a traditional epidemic in the wild reservoir host, the model retains the capacity to implement seasonal variation or a constant force of infection in that species by changing the equations describing the wild compartment. However, our model is intended to codify the idea of an intermediate host mathematically and therefore does not focus on a particular infectious disease. By changing its parameters, this model can be applied to study any zoonosis that passes through an intermediate host population, and its results are general to that theory. 

\subsection{Outline}

This paper investigates two questions: how to model adaptation of a zoonotic pathogen to a human-transmissible form in an intermediate host population and what effects these interspecies dynamics have on the epidemic in humans. Section 2 introduces the model; Section 3 analyzes its mathematical qualities, including its equilibria and $R_0$; Section 4 provides numerical simulations; and Section 5 suggests directions for future research. We find that completely accounting for the spillover and interpopulation dynamics exhibited by emerging zoonoses links human populations to animal ones more deeply than previously thought. With nonzero contact rates between species and a nonzero mutation rate in an intermediate host, a zoonotic pathogen can establish itself in humans even if it fails to take hold in animal hosts or achieve an $R_0 > 1$ in the human compartment, refuting the transmissibility framework (\citep{morse2012prediction}, \citep{wolfe2005bushmeat}, and \citep{lloyd2009epidemic}) that currently forms the basis for classification of emerging zoonoses. This paper introduces a theory of spillover through an intermediate host species that can be modified to study any zoonosis that exhibits this behavior, and sounds an alarm for researchers and policymakers by showing that zoonotic epidemics can persist in human populations under less stringent conditions than previously assumed.

\section{The Model} \label{ch:model}

The traditional susceptible, infected, recovered (SIR) model originally developed by Kermack and McKendrick (\citep{kermack1927contribution}, \citep{kermack1932contributions}, \citep{kermack1933contributions}), shown below, accurately reproduces the standard epidemic curves for an infectious disease, and forms the basis for many different epidemiological models.

\begin{align*}
\frac{dS}{dt} &= -\beta S I, \\
\frac{dI}{dt} &= \beta S I - \gamma I, \\
\frac{dR}{dt} &= \gamma I.
\end{align*}

Here $S$, $I$, and $R$ stand for the proportion of individuals that are susceptible, infected, and recovered, respectively. This deterministic framework depends on the transmission rate $\beta$ and the recovery rate $\gamma$, parameters specific to the disease. This, the simplest version, considers one disease that confers lifelong immunity spreading within a closed, constant population of one species. There are no equilibria other than the disease-free state, as there is no external force of infection or influx of susceptible individuals. The basic reproductive number $R_0$ for this simple SIR model is $$R_0 = \beta/\gamma.$$ 

To model a disease over a longer time frame, vital dynamics modeling birth ($b$) and mortality ($m$) rates are introduced: 
\begin{align*}
dS/dt &= b -\beta S I - m S,\\
dI/dt &= \beta S I - \gamma I - m I,\\
dR/dt &= \gamma I - m R.
\end{align*}

This framework has a force $b$ constantly introducing new susceptible individuals, and the basic reproductive number for this model is
$$ R_0 =  \frac{b\beta}{m(m+\gamma)}.$$ 
The disease-free equilibrium is 
$$(S^*, I^*, R^*) = (\frac{b}{m},\, 0,\, 0),$$ 
and the endemic equilibrium is 
$$(S^*, I^*, R^*) = (\frac{m+\gamma}{\beta},\, \frac{m}{\beta}(R_0-1),\, \frac{\gamma}{\beta}(R_0-1)).
$$

To our knowledge, the main class of SIR models that include two or more species are those that consider vector-borne illnesses. However, since a vector-borne disease must infect both its host species (rather than opportunistically jumping to a new species) and follows set steps in its life cycle in both (rather than unpredictably mutating in a new host), a vector-borne SIR model merely adds more compartments for the pathogen to run through. Unlike vector-borne diseases such as malaria (see \cite{florens2002proteomic} and \cite{chang2013malaria}), an emerging zoonosis does not need to infect another species as part of its life cycle. Instead, it opportunistically infects animals similar enough to its reservoir host, and$-$in the pattern of transmission considered here$-$mutates to a human-to-human transmissible form only if given the opportunity. Dengue, which spreads between mosquitoes and humans, is another example of a vector-borne disease, and its analysis draws useful parallels with the type of pathogen behavior modeled here. \cite{andraud2012dynamic}, a review paper of deterministic models of dengue, notes that the disease dynamics among the vector population are frequently simplified to a mere force of infection for the human one, since the disease does not evolve within the vector species. In contrast, a zoonosis model must consider the disease dynamics in its nonhuman compartments, since these dynamics determine whether the pathogen reaches humans at all. Attempts have been made to model zoonotic spillovers (\cite{lloyd2009epidemic}, \cite{allen2012mathematical}, \cite{hussaini2017mathematical}), but without incorporating changes in the pathogen's ecology over the course of an epidemic, these models are mathematically indistinguishable from those modeling a vector-borne disease with more hosts or a multispecies model. While a sizeable literature exists on mathematical models of vectorborne diseases, and this class of pathogen provides a useful comparison for the type of behavior modeled here, no model captures the unintentional opportunism of zoonoses or incorporates selective pressure on viruses \citep{allen2012mathematical}.

 To model this behavior, we create three compartments, representing the pathogen's wild reservoir host, an intermediate host assumed to be a domestic animal, and humans. The wild, domestic and human populations are each modeled by a SIR system with vital dynamics and linked by transmission routes. An infected wild host can pass the disease to a susceptible domestic animal at a transmission rate $p_d$, and an infected domestic animal can pass the human-transmissible strain of the disease to a human at a rate $p_h$. Finally, the model incorporates the hallmark of an emerging zoonosis: the pathogen's ability to mutate to a human-transmissible strain while circulating in a domestic host. To model this phenomenon, we introduce a category $T$ (transmissible) for domestic animals in which the zoonosis has mutated to a human-transmissible form. This mutation happens at a rate $\mu$ in infected domestic animals, who then transition from the original infected category to the transmissible one and can infect other susceptible domestic animals with the new, human-transmissible strain. The full system of 10 ordinary differential equations is shown in Table~\ref{tab:equations}, with subscripts indicating the species (wild, domestic, or human) to which the compartment belongs. Figure~\ref{fig:schematic} provides a representation of the connections between populations, and Table~\ref{tab:params} gives the definition of each variable.

\begin{figure}[!ht]
\centering
\includegraphics[scale=0.4]{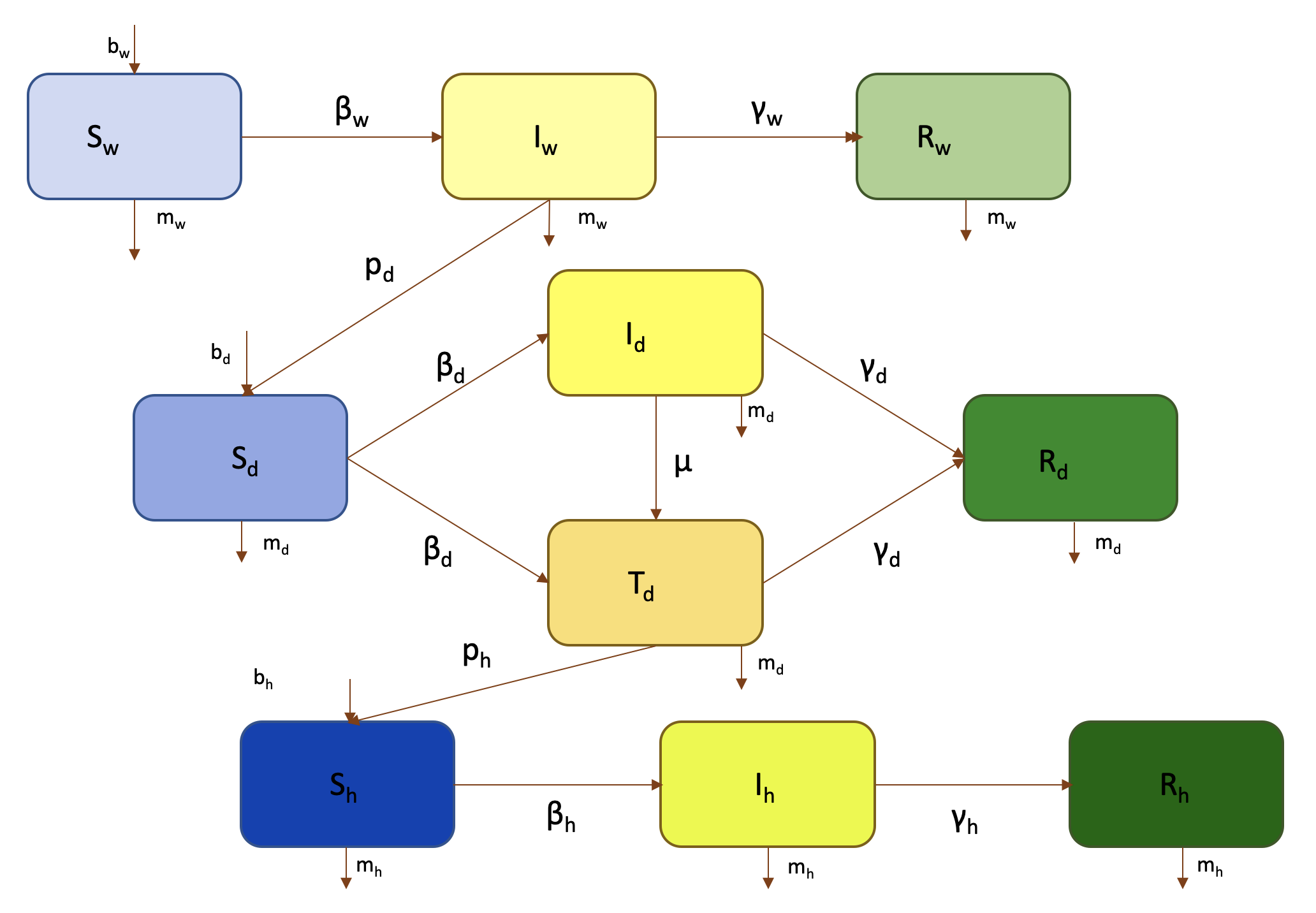}
\caption{A representation of the model. Model parameters are summarized 
in Table~\ref{tab:params}.}
\label{fig:schematic}
\end{figure}

\begin{table}[!ht]
    \centering
    \begin{tabular}{|c|l|}
    \hline
   \multirow{3}{*}{Wild} 
    & $dS_w/dt = b_w - \beta_w S_w I_w - m_w S_w$ \\ 
    & $dI_w/dt = \beta_w S_w I_w - \gamma_w I_w -m_w I_w$ \\
    & $dR_w/dt = \gamma_w I_w - m_w R_w$  \\
   
    \hline
      \multirow{4}{*}{Domestic} 
    & $dS_d/dt = b_d - \beta_d S_d I_d -  p_d S_d I_w - \beta_d S_d T_d - m_d S_d$ \\
   &  $dI_d/dt = \beta_d S_d I_d + p_d S_d I_w - \mu I_d- \gamma_d I_d -m_d I_d$  \\
   &  $dT_d/dt = \mu I_d + \beta_d S_d T_d - \gamma_d T_d -m_d T_d$  \\
   & $dR_d/dt = \gamma_d I_d + \gamma_d T_d -m_d R_d$ \\
    \hline
      \multirow{3}{*}{Humans} & 
     $dS_h/dt = b_h - \beta_h S_h I_h - p_h S_h T_d - m_h S_h$  \\
   &  $dI_h/dt = \beta_h S_h I_h + p_h S_h T_d - \gamma_h I_h - m_h I_h$  \\
   & $dR_h/dt = \gamma_h I_h - m_h R_h$ \\
    \hline
    \end{tabular}
    \caption{ODE systems of our model with three host compartments (species), composed of wild reservoir hosts, intermediate domestic animal hosts, and human hosts.   }
    \label{tab:equations}
\end{table}

\begin{table}[!ht]
    \centering
    \begin{tabular}{| c  l |}
    \hline
   
    $S_i$ & susceptible individuals of species $i$\\
    $I_i$ & infected individuals of species $i$\\
    $T_d$ & intermediate hosts infected with human-transmissible strain\\
    $R_i$ & recovered individuals of species $i$\\
    
    $\beta_i$ & transmission rate among species $i$\\
    $\gamma_i$ & recovery rate among species $i$\\
    $b_i$ & birth rate among species $i$\\
    $m_i$ & natural mortality rate among species $i$\\
    
    $p_d$ & transmission rate from reservoir to intermediate hosts\\
    $p_h$ & transmission rate from intermediate hosts to humans \\
    $\mu$ & mutation rate of the pathogen in the intermediate host population\\
    
    \hline
    \end{tabular}
    \caption{Parameter definitions.}
    \label{tab:params}
\end{table}

\subsection{Model Assumptions}

We make several assumptions to clarify the essential dynamics of the system. Firstly, we equate the domestic animal recovery and transmission rates for both strains of the pathogen; the human-transmissible strain is different from the wild one only in that its transmission rate in humans is nonzero. We further assume that the population of each compartment is constant over the course of the simulation, with each species' vital dynamics set at replacement rates, and thus calculate the proportion of susceptible, infected, and recovered animals in each species rather than the raw numbers present in each category. To maintain a focus on population biology and the potential for the spread of disease from infected individuals, we do not consider disease-induced mortality. Finally, only domestic animals infected with the $T$ strain can pass the disease to humans, although both strains circulate in the domestic population. The model does not account for coinfection in a domestic animal, since an individual infected with both strains is still capable of starting a human epidemic and is thus counted in the $T$ category.

We intend this model to provide a general framework that can be modified to fit any zoonosis with an intermediate host, and provide the analysis in Section 3 with the goal of starting such a theory. However, to provide a baseline for the numerical simulations in Section 4, we use parameters corresponding to highly pathenogenic avian influenza. Table~\ref{tab:paramvals} provides the baseline values and the sources used in our examples. 

\medskip
\begin{table}[!ht]
    \centering
    \begin{tabular}{|c | c | c|}
\hline
Parameter & Value & Source \\
\hline
    initial $S_w$ & 0.5 & \cite{singh2018assessing}\\
    initial $I_w$ & 0.5 & \cite{singh2018assessing}\\
    $p_d$ & 0.51 & \cite{singh2018assessing} figure 1\\
    $\beta_d$   & 0.89 & \cite{henaux2010model} table 1 for wild birds \\
    $\gamma_d$ & 0.981 & \cite{henaux2010model} table 1 \\
    $b_d$ & 1  & assumed \\
    $m_d$ & 1 & assumed \\
    $p_h$ & 0.207 & \cite{xiao2014transmission} \\
    $\beta_h$ & 0.078 & \cite{xiao2014transmission} \\
    $\gamma_h$ & 0.091 & \cite{xiao2014transmission} \\
    $b_h$ & 0.0118 & \href{https://www.cdc.gov/nchs/fastats/births.htm}{CDC} \\
    $m_h$ & 0.009 & \href{https://www.cdc.gov/nchs/nvss/deaths.htm}{CDC} \\
    $\mu$ & 0.499 & \cite{singh2018assessing} figure 3\\
\hline
\end{tabular}
    \caption{Parameter values and sources for the model. Due to a lack of data for transmission parameters in wild animals, we assume $\beta_w, \gamma_w, b_w$, and $m_w$ to be equivalent to their counterparts in domestic animals.}
    \label{tab:paramvals}
\end{table}

Reflecting the lack of data for zoonoses over their entire range of species, the sources used for these parameters reflect different strains of avian influenza. \cite{bett2014transmission} calculates the transmission rate of H5N1 in Nigeria, while \cite{xiao2014transmission} cites information about H7N9 in China. The values are also attained using different data-gathering practices: \cite{singh2018assessing} surveyed experts in Australian avian influenza for their assessment of the probability of domestic poultry becoming infected with low pathogenic avian influenza from wild birds, as well as that strain mutating to higly pathenogic avian influenza (HPAI) on a farm, while \cite{henaux2010model} provides a review of some HPAI parameters. As stated in Table \ref{tab:paramvals}, we could not find a source for transmission parameters among wild birds, and thus assume the disease parameters in that species to be equivalent to those in domestic poultry. The variety and inconsistency of these sources reflects the need for more data and research into the actual effects of particular zoonoses. Although it is crucial for public health interventions based on a mathematical model to know the accuracy of each parameter, their specific values are relatively unimportant for the theoretical results presented here, as the global analysis of the system holds for all parameter values, and are accordingly not the focus of this work.

\subsection{Methods}

To obtain the equilibria for the system, we set each of the 10 equations of Table~\ref{tab:equations} to 0 and solve for the population variables. We further analyze the stability of each equilibrium using the system's Jacobian about the point and establish the importance of the model's basic reproduction number as a threshold condition. The methods we use to analyze the model's $R_0$ are based on the next-generation matrix technique given by \cite{diekmann2009construction} and \cite{van2002reproduction}. This method defines $R_0$ in a compartmental model, where it has been proven to remain a threshold condition for the stability of equilibria \citep{van2002reproduction}. This approach is similar to that used to model the spread of avian influenza in farm and market populations of domestic poultry \citep{li2018dynamics}; to analyze the effect of different growth laws in the avian population on the spread of avian influenza \citep{liu2017nonlinear}; to give a model of a vector-host system for leishmaniasis \citep{hussaini2017mathematical}; to analyze SEIR models \citep{khan2018complex}; and to analyze models with vaccination \citep{anguelov2014backward}. Our work thus uses established mathematical epidemiology techniques to analyze a new model of infectious disease dynamics, extending the preexisting SIR framework to study the spillover effect of a zoonosis with an intermediate host. The model's key innovations are linking three species together based on their proximity to humans and distinguishing between human-transmissible and non-human-transmissible strains of the pathogen, as no previous models simulate either intermediate hosts for zoonoses or a mutation to a human-transmissible form during the course of the epidemic in animals to study the entire range of an emerging infectious zoonosis. 

\section{Analysis} \label{ch:analysis}

In this section, we analyze the mathematical qualities of the model, proving that a unique endemic equilibrium exists by analyzing each species compartment. We further show that the stability of each equilibrium depends on the system's $R_0$ and distinguish between the importance of intraspecies parameters$-$the transmission ($\beta)$ and recovery ($\gamma$) rates of a species, as well as its birth and mortality rate ($b,m$)$-$and interspecies parameters governing connections between species $-$the contact rates $p_d$ and $p_h$, as well as the rate of mutation $\mu$ to a human-transmissible form. We show that, if there is a nonzero number of infected wild animals, only the second type of parameters can alter the global stability of the system.

\subsection{The Wild Compartment}

The equilibrium states $(S_w^*, I_w^*, R_w^*)$ in the wild compartment satisfy the following equations:
 \begin{eqnarray}
     b_w - \beta_w S_w^* I_w^* - m_w S_w^* &= 0,\label{sirw1} \\
    \beta_w S_w^* I_w^* - \gamma_w I_w^* -m_w I_w^* &= 0,\label{sirw2} \\
    \gamma_w I_w^* - m_w R_w^* &= 0\label{sirw3}.
\end{eqnarray}
By summing \eqref{sirw1}, \eqref{sirw2}, and \eqref{sirw3}, we obtain the total abundance of wild animals in equilibrium, $b_w/m_w$.

\begin{theorem}
    There is one disease-free equilibrium, $E_f^w$, at 
    $$(S_w^*, I_w^*, R_w^*) = \left(\frac{b_w}{m_w}, 0, 0\right),$$ 
    and a unique endemic equilibrium, $E_e^w$, at 
$$(S_w^*, I_w^*, R_w^*) = \left(\frac{m_w+\gamma_w}{\beta_w}, \frac{b_w-m_wS_w^*}{\beta_wS_w^*}, \frac{\gamma_wI_w^*}{m_w}\right)$$
\end{theorem}

\begin{proof}
    Factoring equation \eqref{sirw2} yields 
    $$I_w^*(\beta_wS_w^*-\gamma_w-m_w) = 0,$$
    which holds either if  $I_w^* = 0$ (case 1) or if $\beta_wS_w^*-\gamma_w-m_w = 0$ (case 2). 

    In the first case, we obtain the disease-free equilibrium by substituting $I_w^* = 0$ into equations \eqref{sirw1} and \eqref{sirw3}, producing equilibrium values of $S_w^* = \frac{b_w}{m_w}$ and $R_w^* = 0$.
    
    The second case holds if $S_w = \frac{\gamma_w+m_w}{\beta_w}$. Substituting this value into equation \eqref{sirw1}, we obtain $I_w^* = \frac{b_w-m_wS_w^*}{\beta_wS_w^*}$. Solving equation \eqref{sirw3} for $R_w$ gives $R_w^* - \frac{\gamma_wI_w^*}{m_w}$. Since $I_w^* > 0$, this case produces an endemic equilibrium.
\end{proof}

It is a basic epidemiological result that a simple SIR model with vital dynamics, such as the system that models the wild compartment here, has $R_0^w = \frac{b_w\beta_w}{m_w(\gamma_w+m_w)}$. We prove the threshold value of $R_0^w$ in the wild compartment by using its Jacobian,

$$J_w = \begin{bmatrix}
-\beta_wI_w-m_w & -\beta_wS_w & 0 \\
\beta_wI_w & \beta_wS_w-\gamma_w-m_w & 0 \\
0 & \gamma_w & -m_w \\
\end{bmatrix}$$

\begin{theorem}
    $E_f^w$ is stable if $R_0 < 1$ and $E_e^f$ is stable if $R_0 > 1$.
\end{theorem}

\begin{proof}
    In the first case, let $R_0^w< 1$. We calculate that 
    
    $$J_w(E_f^w) = \begin{bmatrix}
    -m_w & -\frac{b_w\beta_w}{m_w} & 0 \\
    0 & \frac{b_w\beta_w}{m_w}-\gamma_w-m_w & 0 \\
    0 & \gamma_w & -m_w \\
    \end{bmatrix} 
    $$

    has eigenvalues $-m_w$ and $\frac{b_w\beta_w}{m_w}-\gamma_w-m_w$. Since $m_w > 0$ by assumption and  $\frac{b_w\beta_w}{m_w} < m_w+\gamma_w$ by the restriction on $R_0$, both eigenvalues are negative and so $E_f^w$ is stable when $R_0^w<1$.
    
    In the second, let $R_0^w > 1$. We have that 
    
    $$J_w(E_e^w) = \begin{bmatrix}
    -\frac{b_w\beta_w}{\gamma_w+m_w} & -\gamma_w-m_w & 0 \\
    \frac{b_w\beta_w}{\gamma_w+m_w}-m_w & 0 & 0 \\
    0 & \gamma_w & -m_w \\
    \end{bmatrix} 
    $$

    with eigenvalues $-\frac{b_w\beta_w}{\gamma_w+m_w}, 0, -m_w$. Since all parameters are positive, these eigenvalues are all negative and thus $E_e^w$ is stable.
\end{proof}

We have thus shown that one disease-free equilibrium and one endemic equilibrium exist among wild reservoir hosts, confirming the importance of $I_w > 0$ as a threshold condition for the spread of the disease.

\subsection{The Domestic Compartment}

In a similar manner, we can analyze the domestic compartment distinctly from the other two species, since interspecies interactions are limited to the force of infection $p_dS_dI_w$ attributed to the wild compartment. Any equilibrium $(S_d^*, I_d^*, T_d^*, R_d^*)$ in this compartment must satisfy the system
\begin{eqnarray}
    b_d - \beta_d S_d^* I_d^* -  p_d S_d^* I_w^* - \beta_d S_d^* T_d^* - m_d S_d^* &= 0,\label{sird1}  \\ 
    \beta_d S_d^* I_d^* + p_d S_d^* I_w^* - \mu I_d^*- \gamma_d I_d^* -m_d I_d^* &= 0, \label{sird2}  \\
    \mu I_d^* + \beta_d S_d^* T_d^* - \gamma_d T_d^* -m_d T_d^* &= 0, \label{sird3}  \\
    \gamma_d I_d^* + \gamma_d T_d^* -m_d R_d^* &= 0. \label{sird4}
\end{eqnarray}
Note that by summing equations \eqref{sird1}-\eqref{sird4}, we obtain the abundance of the domestic compartment at equilibrium, $b_d/m_d$. Since this compartment is subject to an external force of infection from the wild compartment, we also note that the existence of a disease-free equilibrium depends on this external influence.

\begin{lemma}
    A disease-free equilibrium, $E_f^d$, in the domestic compartment, $$(S_d^*, I_d^*, T_d^*, R_d^*) = \left(\frac{b_d}{m_d}, 0, 0, 0\right),$$ is only possible if $I_w = 0$ or $p_d = 0$.
\end{lemma}

\begin{proof}
    Let $I_w > 0$ for some value of $t$ and $p_d \neq 0$, and assume that a disease-free equilibrium exists with $I_d^* = 0$. We note that a disease-free equilibrium requires a nonzero proportion of susceptible individuals, so $S_d^* > 0$ as well. Substituting these values into equation \eqref{sird2}, we obtain $p_dS_d^*I_w^* = 0$, a contradiction with our assumptions. Therefore $I_d^* \neq 0$. Substituting this value into equation \eqref{sird3}, we also obtain $T_d^* \neq 0$. However, these condition imply that there are infected individuals of both types in the domestic animal population, a contradiction with our assumption that we are analyzing a disease-free equilibrium. By contradiction, any disease-free equilibrium in the domestic compartment must have either $I_w = 0$ or $p_d = 0$ in the complete system.
\end{proof}

Note that this result is deeper than one about the equilibrium state of $I_w$: if $I_w > 0$ at any time over the course of the epidemic, even if the disease later vanishes from the wild population, the pathogen will spread to the domestic species.

Assuming that there is a force of infection from the wild reservoir hosts, we analyze the possible endemic equilibrium values and show that there is a unique possibility in this compartment as well.

\begin{theorem}
    There is only one admissible endemic equilibrium $E_e^d$ in the domestic compartment. At this equilibrium, we have $S_d^* < \min\{b_d/(m_d + p_dI_w^*), (\gamma_d + m_d)/\beta_d\}$.
\end{theorem}

\begin{proof}
    Adding equations \eqref{sird1}-\eqref{sird3}, we obtain
    
    \begin{equation*}
    b_d-m_dS_d^*-(\gamma_d+m_d)I_d^*-(\gamma_d+m_d)T_d^* = 0.  \tag{$\star_1$}
    \end{equation*}

    Further, from equation \eqref{sird3}, we isolate

    \begin{equation*}
    I_d^* = \frac{1}{\mu}(\gamma_d+m_d-\beta_dS_d^*)T_d^*. \tag{$\star_2$} 
    \end{equation*}

    Substituting $\star_2$ into $\star_1$, we get
    
    \begin{equation*}
    T_d^* = \frac{b_d-m_dS_d^*}{(\gamma_d+m_d)+\frac{1}{\mu}(\gamma_d+m_d)(\gamma_d+m_d-\beta_dS_d^*)}. \tag{$\star_3$}
    \end{equation*}
    
        Since the quantities $I_d^*$ and $T_d^*$ are both nonnegative, it follows immediately that the equilibrium value $S_d^*$ must have a natural upper bound: 
\begin{equation}
S_d^* \le \frac{\gamma_d + m_d}{\beta_d} \le \frac{b_d}{m_d}. \label{sdbound}
\end{equation}
Hence we can confirm the denominator in ($\star_3$) is strictly positive.

    We also note that from \eqref{sird4}, we have $R_d^* = \frac{\gamma_d(I_d^*+T_d^*)}{m_d}$. We thus calculate the equilibrium values $S_d^*$ by substituting $\star_2, \star_3$ into \eqref{sird1} and obtain, after rearranging:
    \begin{align}
    b_d-p_dS_d^*I_w^*-m_dS_d^* &= \beta_dS_d^*(I_d^*+T_d^*) =\beta_dS_d^*(1+\frac{1}{\mu}(\gamma_d+m_d-\beta_dS_d^*))T_d^* \nonumber\\
   &  = \frac{\beta_dS_d^*(1+\frac{1}{\mu}(\gamma_d+m_d-\beta_dS_d^*))(b_d-m_dS_d^*)}{(\gamma_d+m_d)+\frac{1}{\mu}(\gamma_d+m_d)(\gamma_d+m_d-\beta_dS_d^*)}\nonumber\\
   & = \frac{\beta_dS_d^*(b_d-m_dS_d^*)}{(\gamma_d+m_d)}.\label{sdeq}
    \end{align}

We note that in the last step of the derivations above we cancel out the common factor $1+\frac{1}{\mu}(\gamma_d+m_d-\beta_d S_d^*) > 0$, which is guaranteed by the inequality ~\eqref{sdbound}. It is easy to observe that there exist at most two possible equilibrium values of $S_d^*$ as the roots of the quadratic equation \eqref{sdeq}, denoted by $S_{d_(1)}^* < S_{d_(2)}^*$, a consequence of our assumption that the transmission and recovery rates of both strains in domestic animals are equal.

We now proceed to prove only the smaller root $S_{d_(1)}^*$ is admissible for the long-term disease dynamics should there be nonzero disease burden (i.e., $I_d^* > 0$ and $T_d^* > 0$) in the domestic compartment. 
    In fact, we can view $S_d^*$ as the fixed point(s) satisfying 
    \[
    f(x) = g(x),
    \]
    where $f(x)$ is a simple linear function, given by
    \[
    f(x) =  b_d-(p_dI_w^*+m_d)x,
    \]
and $g(x)$ is a quadratic function, given by  
\[
g(x) = \frac{\beta_d x (b_d-m_d x)}{(\gamma_d+m_d)}.
\]

We can show that $f(0) = b_d > 0 = g(0)$, $f(b_d/(m_d+p_dI_w^*)) = 0 < g(b_d/(m_d+p_dI_w^*))$, and $0 > f(x) > g(x)$ for sufficiently large $x$. Furthermore, as both $f$ and $g$ are smooth continuous functions, according to the intermediate value theorem, there must exist two fixed points satisfying $f(x) = g(x)$, $S_{d_(1)}^* \in (0, b_d/(m_d+p_dI_w^*))$ and $S_{d_(2)}^* \in (b_d/(m_d+p_dI_w^*), \infty)$ (as also illustrated in Figure~\ref{fig:domesticEquilibriaGraph}).

We have $b_d-p_dS_d^*I_w^*-m_dS_d^* = \beta_dS_d^*(I_d^*+T_d^*) > 0$, for nonzero disease burden $I_d^* > 0, T_d^* > 0$. Hence we must have $S_d^* < b_d/(m_d+p_dI_w^*)$. So we complete our proof that only the smaller root $S_{d_(1)}^*$ is admissible as the unique endemic equilibrium in the domestic compartment.

    \begin{figure}
    \centering
    \includegraphics[scale=0.3]{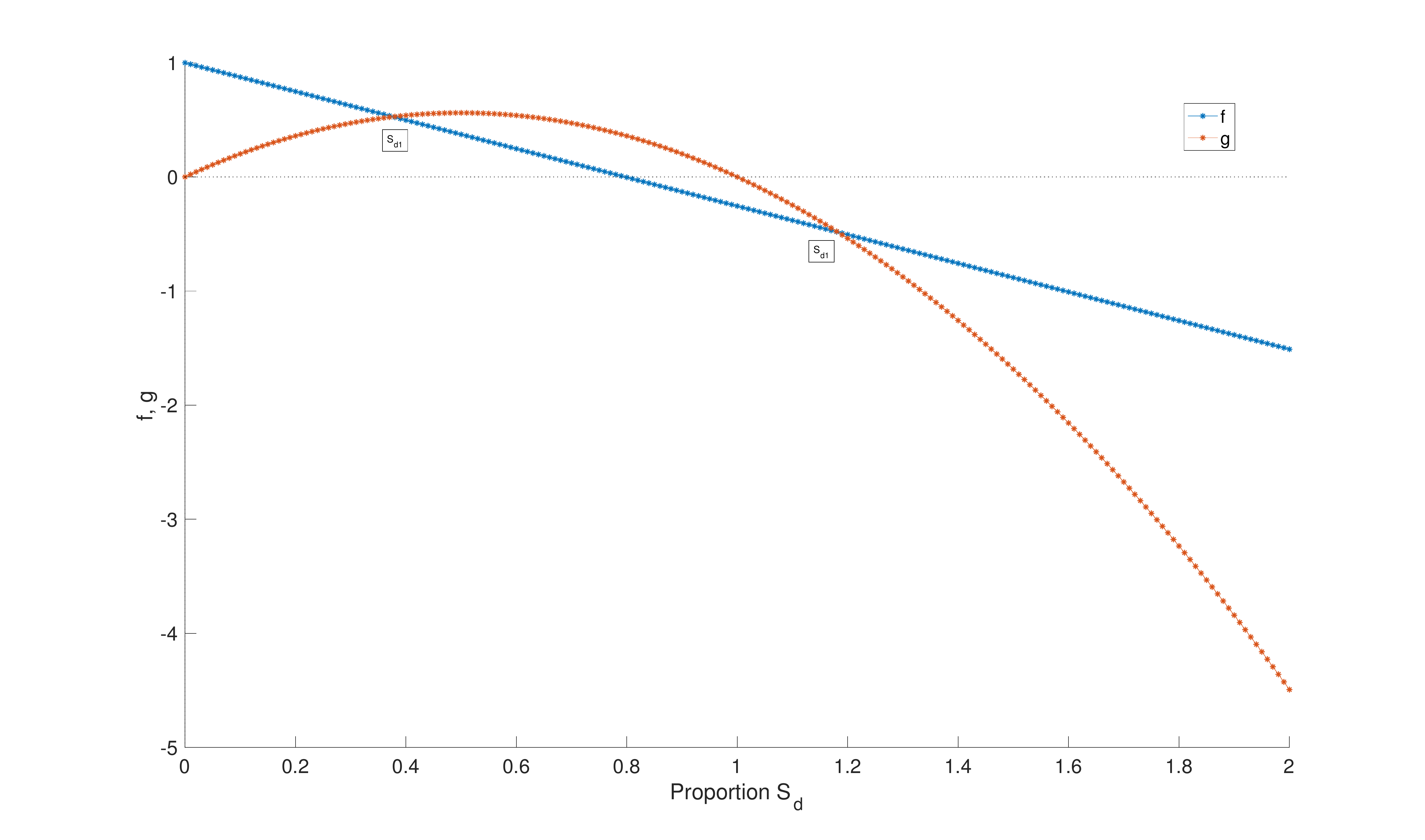}
    \caption{The graphs of $f(S_d^*)$ and $g(S_d^*)$; the $x$-coordinates of their intersection points give the proportion of domestic animals infected at the equilibria.}
    \label{fig:domesticEquilibriaGraph}
    \end{figure}

\end{proof}

In line with our proof above, it is easy to directly check the discriminant  of the quadratic equation in \eqref{sdeq} is positive:
\[
\Delta = [\beta_d b_d + (\gamma_d + m_d)(p_dI_w^* + m_d)]^2 - 4\beta_d m_d b_d(\gamma_d + m_d) > 0.
\]
Then we obtain the endemic equilibrium for $S_d^*=S_{d_(1)}^*$:
\[
S_{d_(1)}^* = \frac{\beta_d m_d + (\gamma_d+m_d)(p_d I_w^* + m_d) - \sqrt{\Delta}}{2\beta_d m_d}.
\]

We have thus shown that, in the domestic compartment, there is one unique endemic equilibrium $E_e^d = (S_d^*, I_d^*, T_d^*, R_d^*)$, where 
\begin{align*}
   S_d^* & = S_{d_(1)}^*, \\
   T_d^* & = \frac{b_d-m_dS_d^*}{(\gamma_d+m_d)+\frac{1}{\mu}(\gamma_d+m_d)(\gamma_d+m_d-\beta_dS_d^*)},\\
   I_d^* & = \frac{1}{\mu}(\gamma_d+m_d-\beta_dS_d^*)T_d^*, \\
   R_d^* & = \frac{\gamma_d(I_d^*+T_d^*)}{m_d}.
\end{align*}

 Further, as we prove above, this equilibrium must be admissible in the presence of a nonzero force of infection at any time from the wild host population ($p_d I_w^* > 0$). 

\subsection{The Human Compartment}

To complete our understanding of the different species involved in the model, we analyze the system of equilibrium equations in the human compartment as follows, 
\begin{eqnarray}
    b_h - \beta_h S_h I_h - p_h S_h T_d - m_h S_h & =& 0, \label{sirh1}\\
    \beta_h S_h I_h + p_h S_h T_d - \gamma_h I_h - m_h I_h &= & 0, \label{sirh2} \\
    \gamma_h I_h - m_h R_h & = &  0.\label{sirh3} 
\end{eqnarray}
Adding \eqref{sirh1} - \eqref{sirh3} gives a total abundance of $\frac{b_h}{m_h}$ in the human compartment at equilibrium. We begin our analysis of this compartment by noting an identical result from the domestic one: a disease-free equilibrium can exist in the human compartment only if the force of infection from domestic animals is zero. 

\begin{lemma}
    A disease-free equilibrium $E_f^h$ in the human compartment, $(S_h^*, I_h^*, R_h^*) = (\frac{b_h}{m_h}, 0, 0)$, is only possible if $T_d = 0$ or $p_h = 0$.
\end{lemma}

\begin{proof}
    In a manner similar to the proof of Lemma 1, let $T_d > 0$ at any time over the course of the model and $p_h \neq 0$, and assume that a disease-free equilibrium exists with $I_h^* = 0$. By equation (9), we obtain $p_hS_h^*T_d^* = 0$, a contradiction with our assumptions and with the fact that $S_h^* \neq 0$ at a disease-free equilibrium. By contradiction, any disease-free equilibrium must have either $T_d = 0$ or $p_h = 0$.
\end{proof}

Thus, in the presence of any force of infection $p_hT_d$ from domestic animals, there must be an endemic equilibrium in the human compartment. We note again that if $T_d > 0$ at any time over the course of the epidemic, even if that population of animals vanishes at equilibrium, it is enough to seed the infection into the human compartment.

\begin{theorem}
    There exists only one admissible endemic equilibrium in the human compartment $E_h^e = (S_h^*, I_h^*, R_h^*)$, where $S_h^*$ is given by the smaller root of the quadratic equation:
    \[
    b_h - p_h S_h^* T_d^* - m_h S_h^* = \beta_h S_h^*(b_h - m_h S_h^*)/(\gamma_h + m_h).
    \]
\end{theorem}

%$S_d^* = \frac{p_hT_d^*+m_h+\frac{b_h\beta}{\gamma_h+m_h}-\sqrt{(p_hT_d^*+m_h+\frac{b_h\beta}{\gamma_h+m_h})^2-4b_h(\frac{m_h\beta_h}{\gamma_h+m_h})}}{2b_h}$.

\begin{proof}
    From equation \eqref{sirh3}, we know 
    \[R_h^* = \frac{\gamma_hI_h^*}{m_h}.\]
    By adding equations \eqref{sirh1} and \eqref{sirh2}, we obtain
    $$ b_h-m_hS_h^*-(\gamma_h+m_h)I_h^* = 0.$$
Accordingly, $I_h^* = \frac{b_h-m_hS_h^*}{\gamma_h+m_h}$. Substituting $I_h^*$ into \eqref{sirh1}, we obtain 
    
    \begin{equation}
    b_h-p_hS_h^*T_d^*-m_hS_h^* = \frac{\beta_hS_h^*(b_h-m_hS_h^*)}{\gamma_h+m_h}.\label{hqua}
    \end{equation}
    Defining the left-hand side as $f(S_h^*) =  b_h-p_hS_h^*T_d^*-m_hS_h^*$ and the right-hand side as $g(S_h^*) =  \frac{\beta_hS_h^*(b_h-m_hS_h^*)}{\gamma_h+m_h}$, as in the proof of Theorem 3, any endemic equilibrium in the human compartment must satisfy the fixed point(s) $f(S_h^*)=g(S_h^*)$ (see Figure~\ref{fig:humanEquilibriaGraph}).
    
    \begin{figure}
    \centering
    \includegraphics[scale=0.3]{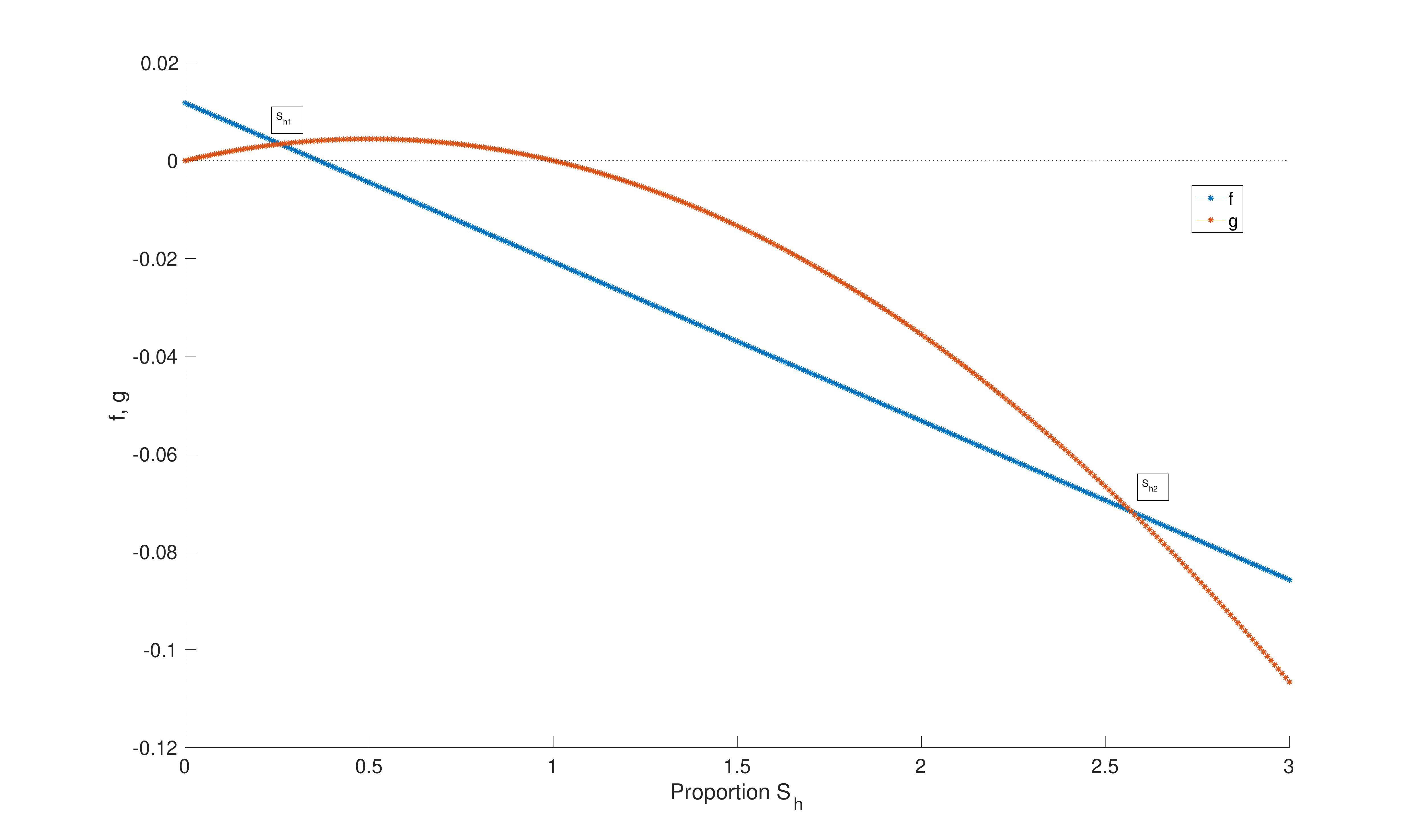}
    \caption{The graphs of $f(S_h^*)$ and $g(S_h^*)$; the $x$-coordinates of their intersection points give the proportion of humans infected at the equilibria.}
    \label{fig:humanEquilibriaGraph}
\end{figure}

    $f$ is a negatively-sloped line with its root at $f_1 = \frac{b_h}{p_hT_d^*+m_h}$. $g$ is a concave parabola with roots at $g_1 = 0$ and $g_2 = \frac{b_h}{d_h}$. We thus have $g_1 < f_1 < g_2$, so there are two intersection points $S_{h_(1)}^* < f_1$ and $ S_{h_(2)}^* > g_2$ (by the intermediate value theorem). However, $g_2$ is the total compartment size in the human compartment; $S_{h_(2)}$ is thus biologically impossible.
    
    Solving the quadratic equation \eqref{hqua} directly and taking its smaller root, we obtain the only viable endemic equilibrium:
    $$S_{h_(1)}^* =
    \frac{\beta_h b_h + (m_h + \gamma_h)(p_hT_d^* + m_d)-\sqrt{[\beta_h b_h + (m_h + \gamma_h)(p_hT_d^* + m_d)]^2-4\beta_h m_h b_h(\gamma_h + m_h)}}{2\beta_h m_h}.$$
   
\end{proof}

We have thus shown that in the presence of a force of infection from domestic animals, there is an admissible endemic equilibrium in the human compartment at $E_h^e = (S_h^*, I_h^*, R_h^*)$, where
\begin{align*}
   S_h^* & = S_{h_(1)}^*, \\
   I_h^* & = \frac{b_h-m_hS_h^*}{\gamma_h+m_h}.\\
   R_h^* & = \frac{\gamma_hI_h^*}{m_h}.
\end{align*}

\subsection{System Stability and Basic Reproduction Number}

As shown in the previous sections, the endemic equilibria in each compartment are unique; under the assumption that there is a nonzero force of infection between species compartments and in the presence of circulating disease in the wild reservoir, we use the results of Theorems 1, 3, and 4 to obtain an endemic equilibrium at
$$
E_e = (S_w^*, I_w^*, R_w^*, S_d^*, I_d^*, T_d^*, R_d^*, S_h^*, I_h^*, R_h^*).
$$
The formula of these expressions can be found in detail above.

More precisely, the existence of such stable endemic disease equilibiria requires an exact condition, that is, the basic reproductive number of the entire system $R_0 > 1$. Otherwise, there can exist a stable disease-free equilibria. We use next-generation approach as detailed in \cite{diekmann2009construction} and \cite{van2002reproduction} to calculate $R_0$ for the system. 

The system's $R_0$ is the spectral radius of $FV^{-1}$, where $F$ describes the rate of appearance of new infections in each compartment of host individuals,
    $$ F = \begin{bmatrix}
        \beta_w S_w & 0 & 0 & 0 \\
        p_d S_d & \beta_d S_d & 0 & 0 \\
        0 & 0 & \beta_dS_d & 0 \\
        0 &0 &p_h S_h & \beta_h S_h\\
    \end{bmatrix}, $$
    
and $V$ describes the rate of transfer of individuals out of each compartment,

    $$V = \begin{bmatrix}
        \gamma_w+m_w & 0 & 0 & 0 \\
        0& \mu+\gamma_d+m_d & 0 & 0 \\
        0 & -\mu & \gamma_d+m_d &0 \\
        0 &0 &0 & \gamma_h+m_h\\
\end{bmatrix}.$$

We thus get 

$$ FV^{-1} = \begin{bmatrix}
        \frac{\beta_wS_w}{m_w(\gamma_w+m_w)}& 0 & 0 & 0 \\
        \frac{p_dS_d}{\gamma_w+m_w}& \frac{\beta_dS_d}{\mu+\gamma_d+m_d}& 0 & 0 \\
        0 & \frac{\mu\beta_dS_d}{(\gamma_d+m_d)(\mu+\gamma_d+m_d)} & \frac{\beta_dS_d}{\gamma_d+m_d} & 0 \\
        0 & \frac{\mu p_hS_h}{(\gamma_d+m_d)(\mu+\gamma_d+m_d)} & \frac{p_hS_h}{\gamma_d+m_d} & \frac{\beta_hS_h}{\gamma_h+m_h}\\
\end{bmatrix}.
$$
 
$R_0$ is then the maximum of the eigenvalues of this matrix, 

$$R_0 = \textrm{max}\{\frac{\beta_wS_w}{\gamma_w+m_w}, \frac{\beta_dS_d}{\mu+\gamma_d+m_d},\frac{\beta_dS_d}{\gamma_d+m_d}, \frac{\beta_hS_h}{\gamma_h+m_h}\}.  $$

At the disease-free equilibrium, we have $S_w = b_w/m_w, S_d = b_d/m_d, S_h = b_h/m_h$. Therefore, the $R_0$ value in a population entirely composed of susceptible individuals is 
\begin{equation}
    R_0 = \textrm{max}\left\{\frac{\beta_wb_w}{m_h(\gamma_w+m_w)}, \frac{\beta_db_d}{m_d(\mu+\gamma_d+m_d)},\frac{\beta_d b_d}{m_d(\gamma_d+m_d)}, \frac{\beta_hb_h}{m_h(\gamma_h+m_h)}\right\}. 
\end{equation}

We further establish that $R_0 > 1$ retains its traditional role as the threshold for determining the viability of an epidemic using an analysis of the system's Jacobian matrix at the disease-free equilibria. Since we are interested only in the total number of infected individuals, we consider the time evolution of disease burden across all three compartments in the form $(I_w, I_d, T_d, I_h)$:
\begin{eqnarray*}
 dI_w/dt & = & \beta_w S_w I_w - \gamma_w I_w -m_w I_w, \\
 dI_d/dt & = & \beta_d S_d I_d + p_d S_d I_w - \mu I_d- \gamma_d I_d -m_d I_d  \\
  dT_d/dt & = & \mu I_d + \beta_d S_d T_d - \gamma_d T_d -m_d T_d  \\
  dI_h/dt & = & \beta_h S_h I_h + p_h S_h T_d - \gamma_h I_h - m_h I_h. 
\end{eqnarray*}

The Jacobian matrix of this system above at the disease-free equilibrium $E_f = (0,0,0,0)$ is

\begin{math}
J(E_f) =
\begin{bmatrix}
\beta_wS_w-(\gamma_w+m_w) & 0 & 0 & 0 \\

p_dS_d & \beta_dS_d-(\mu+\gamma_d+m_d) & 0 & 0 \\

0 & \mu & \beta_dS_d-(\gamma_d+m_d) & 0 \\

0 & 0 & p_hS_h & \beta_hS_h-(\gamma_h+m_h) \\
\end{bmatrix}.
\end{math}\\

We first establish results on the stability of the disease-free equilibrium $E_f$.

\begin{theorem}
    $E_f$ is asymptotically stable if $R_0 < 1$ and unstable if $R_0 > 1$.
\end{theorem}

\begin{proof}
At $E_f = (0, 0, 0, 0)$, we have $S_w = b_w/m_w, S_d = b_d/m_d, S_h = b_h/m_h$. Thus, $J(E_f) = $
$$ 
\begin{bmatrix}
\beta_w b_w/m_w -(\gamma_w+m_w) & 0 & 0 & 0 \\

p_dS_d & \beta_d b_d/m_d-(\mu+\gamma_d+m_d) & 0 & 0 \\

0 & \mu & \beta_d b_d/m_d -(\gamma_d+m_d) & 0 \\

0 & 0 & p_hS_h & \beta_h b_h/m_h-(\gamma_h+m_h)
\end{bmatrix}$$.

If $R_0 < 1$, then the diagonal entries of $J(E_f)$, which are actually the eigenvalues of the Jacobian matrix $J(E_f)$, are strictly negative. Thus $E_f$ is asymptotically stable if $R_0 < 1$. If $R_0 > 1$, then at least one of the eigenvalues of the Jacobian matrix $J(E_f)$ (its diagonal entries) is strictly positive. Therefore $E_f$ is unstable if $R_o > 1$.

\end{proof}

We then establish similar results for the endemic equilibrium. 

\begin{theorem}
    $E_e$ is asymptotically stable if $R_0 > 1$ and unstable if $R_0 < 1$.
\end{theorem}

\begin{proof}
According to our analysis of endemic equilibria above, we have 
\begin{eqnarray*}
S_w^* & = & \frac{\gamma_w + m_w}{\beta_w}, \\
S_d^* & < & \frac{\gamma_d + m_d}{\beta_d},\\
S_h^* & < & \frac{\gamma_h + m_h}{\beta_h}.
\end{eqnarray*}

Substituting into the Jacobian matrix $J(E_e)$ the values for $S_w^*, S_d^*, S_h^*$,  the diagonal entries of $J(E_e)$ are either zero or negative. Therefore $E_e$ is asymptotically stable if $R_0 > 1$. Similarly, we can prove  $E_e$ is unstable if $R_0 < 1$.

\end{proof}

Thus $R_0$ retains its role as the threshold condition for an epidemic.

 \subsection{Threshold Parameters}

The results above replicate the standard epidemiological finding that $R_0$ is a threshold condition for the system, but the interspecies connections in this model allow us to establish a more detailed result. By Theorem 2 and Lemmas 1 and 2, the stability of the equilibria depends on $p_d$, $p_h$ and $\mu$, the same parameters which control $S_w, S_d,$ and $S_h$ in the calculation of $R_0$. Indeed, only these parameters determine the results of the epidemic.

\begin{theorem} 
    In the presence of a nonzero number of infected wild animals, $E_e$ is stable if and only if $p_d, \mu, p_h > 0$.
\end{theorem}

\begin{proof}
    ($\Rightarrow$) If the disease-free equilibrium is stable, Lemmas 1 and 2 show that $p_d, I_w, p_h, T_d > 0$. Considering equation (6), the only way to obtain $T_d > 0$ is to have $\mu > 0$ as well.  
    
    ($\Leftarrow$) If $p_d, p_h, \mu > 0$, a nonzero proportion $I_w$ creates a positive force of infection in equation (5), and since $\mu > 0$ there is a positive force of infection in equation (6) as well. With $T_d^* > 0$ and $p_h > 0$, there is a positive force of infection in equation (9), and so $I_h^* > 0$, creating an endemic equilibrium in the human compartment. If all of these conditions hold, regardless of the values of $\beta_i$ or $\gamma_i$ in any species, there is a nonzero, constant force of infection for each species and so the system is forced into an endemic equilibrium.
\end{proof}

\subsection{Conclusion}

This model has one disease-free equilibrium and one endemic equilibrium, whose stability depends on $p_d, p_h$ and $\mu$. Isolating these parameters thus provides suggestions for possible interventions. The results of Theorem 7, in particular, show that while many parameters of the model can be changed by human intervention$-\beta_d$ could be lowered by increasing biosecurity on farms for domestic animals, for example, while much of public health and medicine offers strategies for changing $\beta_h$ and $\gamma_h-$the only effective route for eliminating the possibility of a zoonotic epidemic in humans is to eliminate contact between species or the possibility of pathogen mutation, an impossible requirement in any real system.

\section{Numerical Simulations} \label{ch:simulations}

To clarify the results of the theoretical analysis presented in section 3, we present simulations of different cases of the model drawn from available data (summarized in Table~\ref{tab:paramvals}).  To elucidate the effects of the interspecies transmission parameters$-p_d$, $p_h$, and $\mu-$we simulate cases where the pathogen fails to establish itself in wildlife, in domestic animals, and in both populations, showing that the human population will still suffer an endemic disease even if animal populations remain relatively unaffected by a brief epidemic.

\subsection{Examples}

We first simulate a zoonosis that establishes endemic equilibria in each host species, using the baseline parameters with $5\beta_w=5\beta_d$ to ensure the spread of the pathogen.

\begin{figure}[!ht]
    \centering
    \includegraphics[scale=0.3]{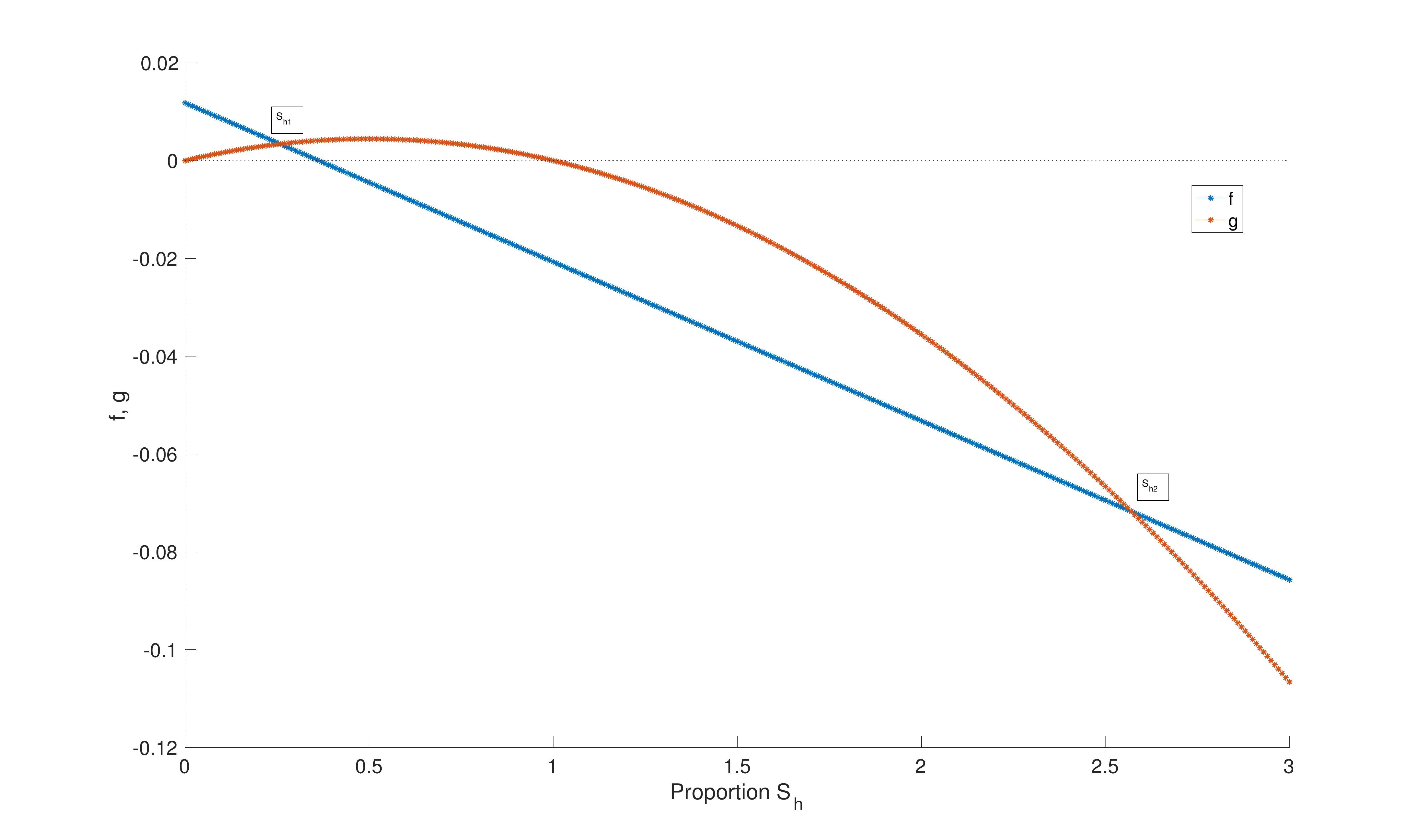}
    \caption{A simulation showing endemic equilibria in each species. Parameter values are as shown in Table~\ref{tab:paramvals}, with $\beta_w$ and $\beta_d$ multiplied by 5 to ensure spread in each compartment.}
    \label{fig:endemic}
\end{figure}

The outbreak shown in figure~\ref{fig:endemic} infects a maximum of 46.33\% of the human population and stabilizes at 19.04\% of the population infected, reaching equilibria in all three species by 150 units of time.

Next, to elucidate the effect of the mutation, we simulate an outbreak that fails to establish itself in the wild population (in this case, this species does not function as a reservoir host). 

\begin{figure}[!ht]
    \centering
    \includegraphics[scale=0.3]{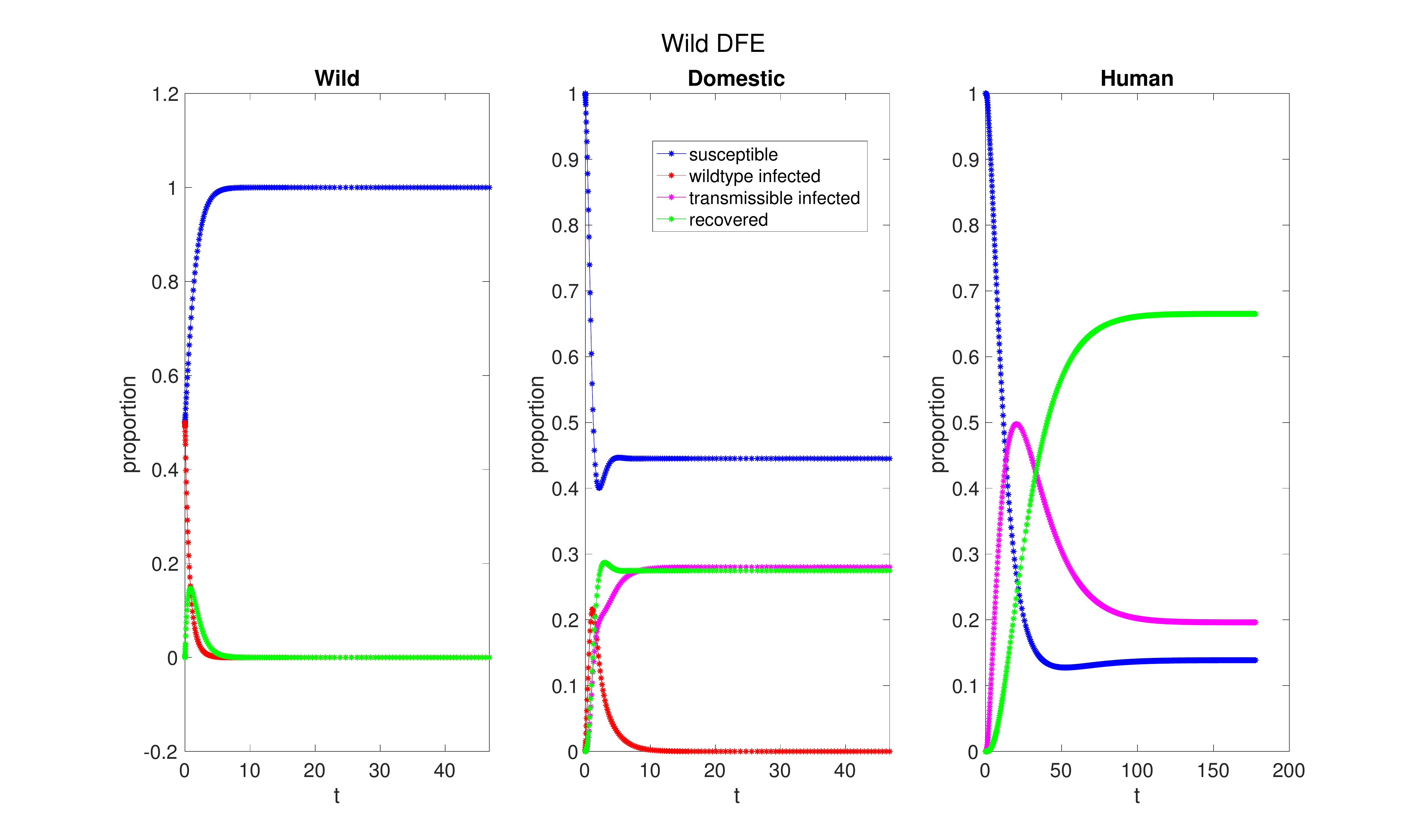}
    \caption{A simulation showing a disease-free equilibrium in the wild reservoir host species that spills over to endemic equilibria in the domestic intermediate host and humans. Parameters are as shown in Table~\ref{tab:paramvals}, except with $\beta_d$ multiplied by 5 to ensure an epidemic in the domestic compartment.}
    \label{fig:wildDFE}
\end{figure}

Figure~\ref{fig:wildDFE} shows that even if the disease fails to persist in its wild reservoir host, it can still become endemic in the human population. A maximum of 49.75\% of the human population was infected, with 19.62\% infected at equilibrium by time 150. This case illustrates that even if the epidemic fails to take hold among wild animals, it can still spread to domestic animals and thus humans, illustrating the importance of $p_d$ as a threshold parameter.

For our final example, we simulate avian influenza mutating from a low-pathogenic to a highly-pathogenic strain in an intermediate host. One of the best-known examples of a zoonosis with an intermediate host, avian influenza spreads from wild birds to domestic poultry to humans, a process for which there is some publicly available data. Seeding the model with the parameters shown in Table~\ref{tab:paramvals} (and assuming that $\beta_w = \beta_d, \gamma_w=\gamma_d$), we obtain the result shown in figure~\ref{fig:AI}.

\begin{figure}[!ht]
    \centering
    \includegraphics[scale=0.3]{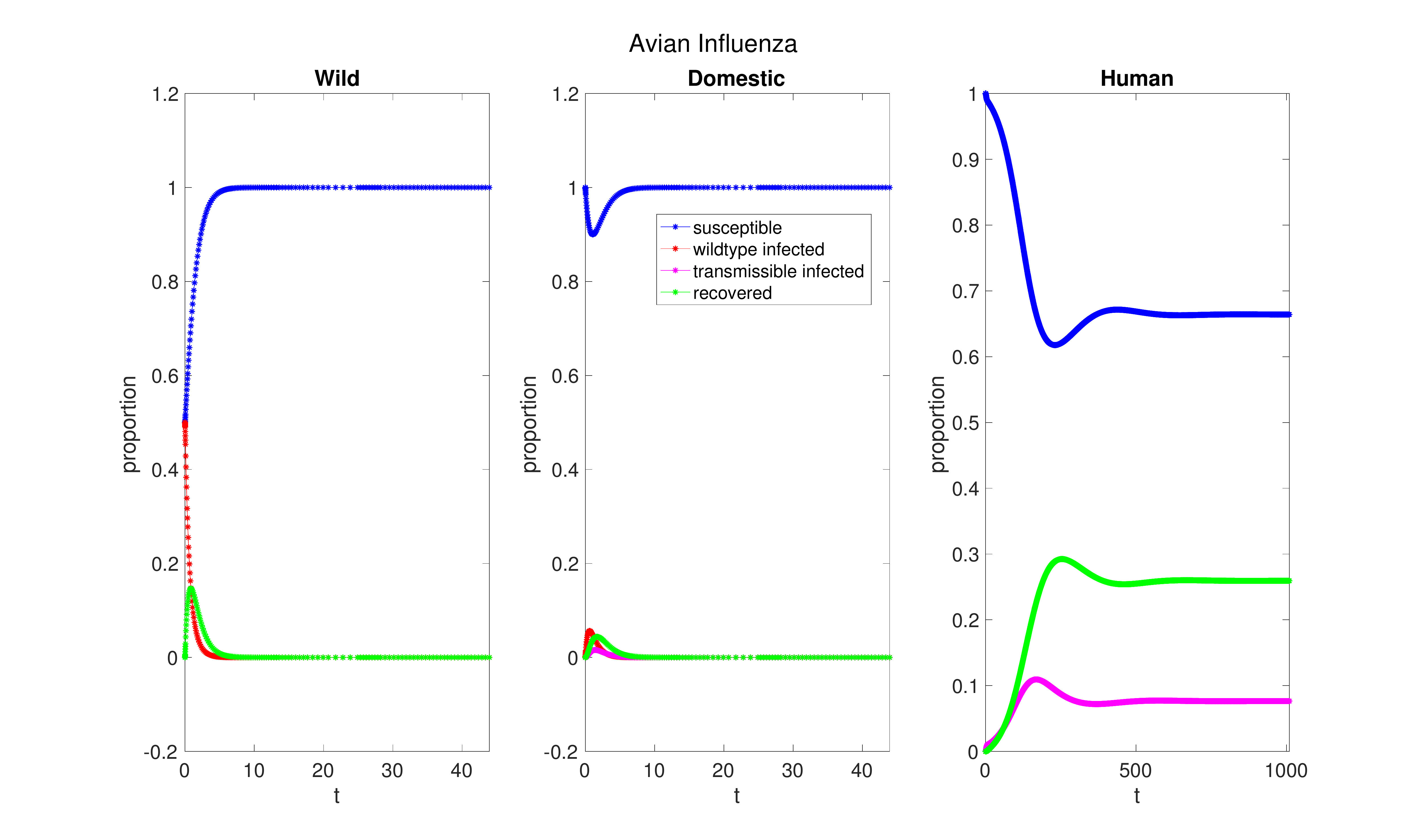}
    \caption{A simulation of low-pathogenic avian influenza mutating to high-pathenogenic avian influenza. Parameters are as shown in Table~\ref{tab:paramvals}.}
    \label{fig:AI}
\end{figure}

This example$-$which uses the most data publicly available$-$shows that even if a pathogen's $R_0$ is less than one in both wild and intermediate hosts, it can still establish itself in the human population. Here, both strains of avian influenza fade in the animal populations while establishing an endemic equilibrium in the human population, with a maximum of 10.94\% and an equilibrium of 7.65\% of the population infected over a time span an order of magnitude larger than that necessary in the previous examples ($t = 2000$, not shown in the figure). This result indicates that the unexpected behavior described in section 3 and modeled above does appear in real epidemics.

The results here are summarized in Table~\ref{tab:results}. The equilibrium proportion of infected humans is highest in row 2 because there are more domestic animals infected with the transmissible strain when the force of infection with the wildtype strain vanishes over time.

\begin{table}[!ht]
    \centering
    \begin{tabular}{|c | c | c| c |}
\hline
Situation & Max $I_h$ & Equilibrium $I_h$ & Time to Equilibrium \\
\hline
    Endemic in all species & 46.33\% & 19.04\% & $10^2$ \\
    Disease-free in wildlife & 49.75\% & 19.62\% & $10^2$ \\
    Avian influenza & 10.94\% & 7.65\% & $10^3$\\
\hline
\end{tabular}
    \caption{A comparison of the maximal and equilibrium values for the percentage of infected humans for each representative strain.}
    \label{tab:results}
\end{table}

These simulations illustrate that with nonzero transmission parameters, an initial infection in an upstream host species will spread to an endemic equilibrium in downstream ones even if the pathogen fails to establish itself in its animal hosts.

\subsection{Effects of Interspecies Transmission Parameters}

    In this section, we evaluate the effect of varying the interspecies transmission parameters $p_d$, $\mu$, and $p_h$ on the equilibrium values $I_d^*$, $T_d^*$, and $I_h^*$ after 3000 units of time, in addition to $\beta_d$ and $\beta_h$ for comparison. To produce the graphs shown below, we vary the parameter in question from 0.01 to 5 (since values of 0, as shown in section 3, inevitably lead to a disease-free equilibrium in the human compartment), with a step size of 0.1, holding the other values constant at the endemic equilibrium parameters detailed in section 4.1. Each simulation is run for 3000 timesteps, to ensure that an equilibrium solution is reached. In the domestic compartment, varying the transmission parameters $p_d$ and $\mu$ can change the relative prevalence of the wildtype and human-transmissible strains, as shown in Figures~\ref{fig:IT_vs_pd} and ~\ref{fig:IT_vs_mu}. (We do not examine the effect varying $p_h$ has on the domestic compartment because that parameter does not appear in the equations governing its behavior.)

    \begin{figure}[!ht]
    \centering
    \includegraphics[scale=0.25]{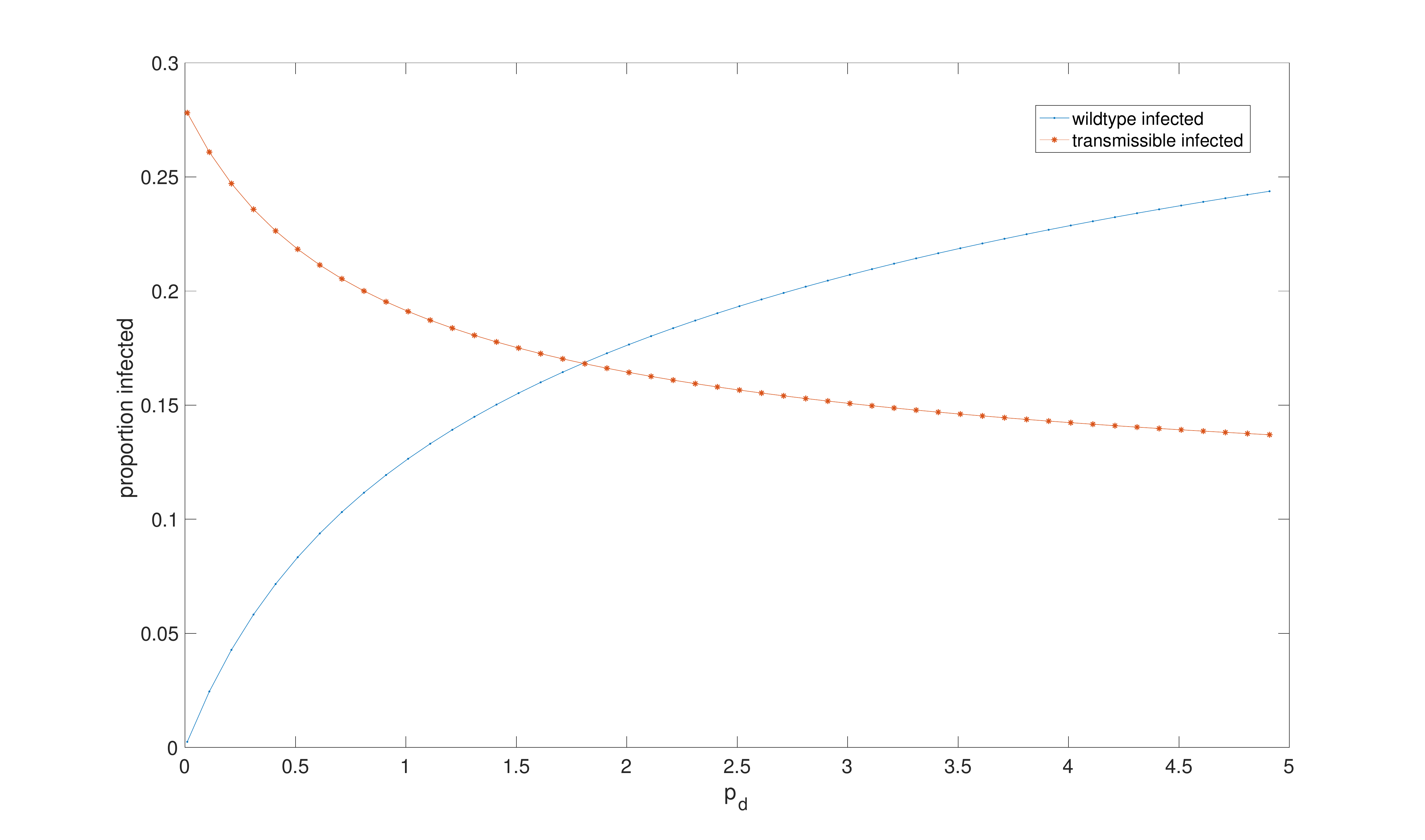}
    \caption{Graphing the proportion of domestic animals infected with the wildtype strain and the human-transmissible strain against $p_d$, the contact rate (spillover rate) between wild animals and domestic ones.}
    \label{fig:IT_vs_pd}
    \end{figure}
    
    \begin{figure}%[!ht]
    \centering
    \includegraphics[scale=0.25]{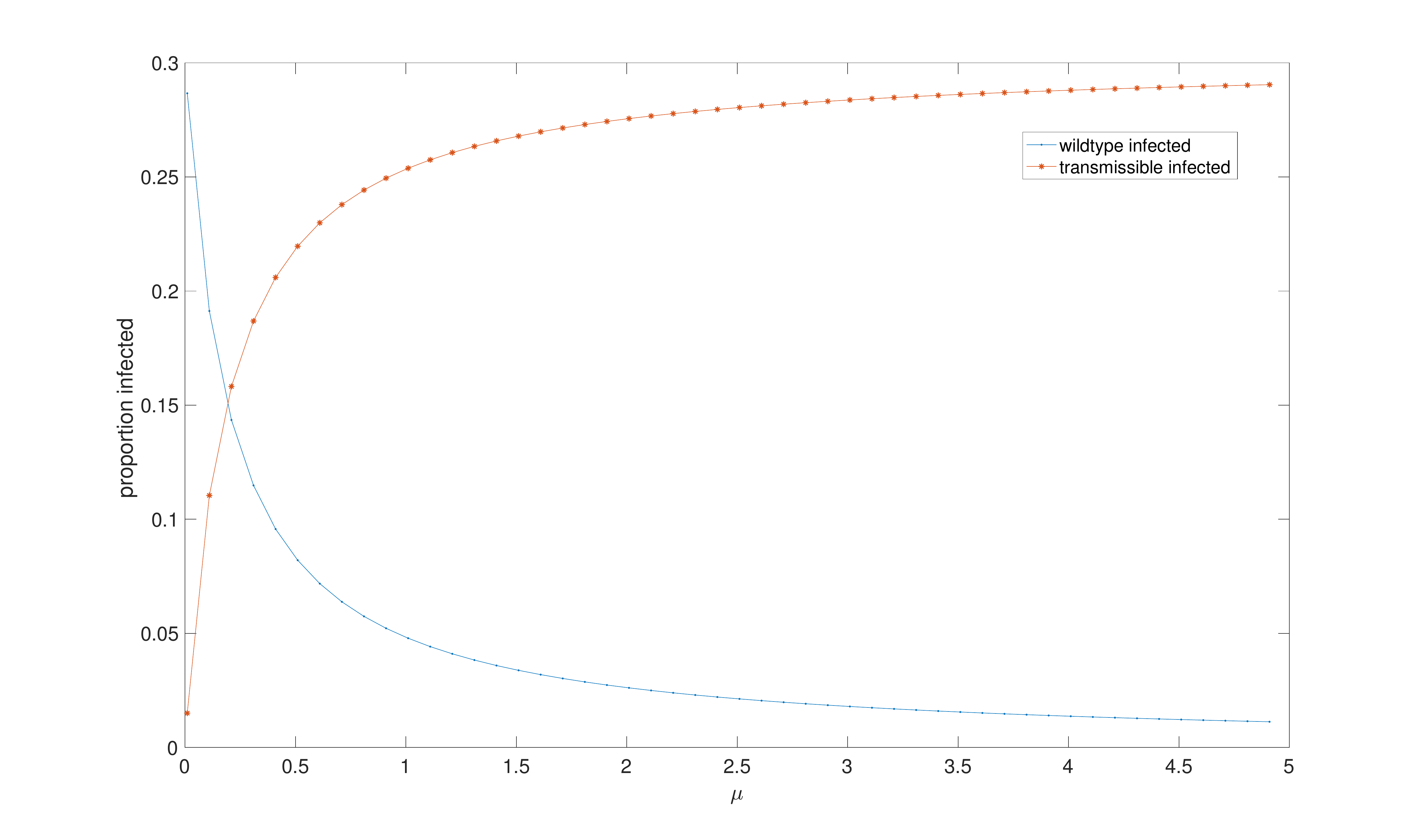}
    \caption{Graphing the proportion of domestic animals infected with the wildtype strain and the human-transmissible strain against $\mu$, the rate of mutation from the wildtype strain to the human-transmissible strain.}
    \label{fig:IT_vs_mu}
    \end{figure}
    
    \begin{figure}%[!ht]
    \centering
    \includegraphics[scale=0.25]{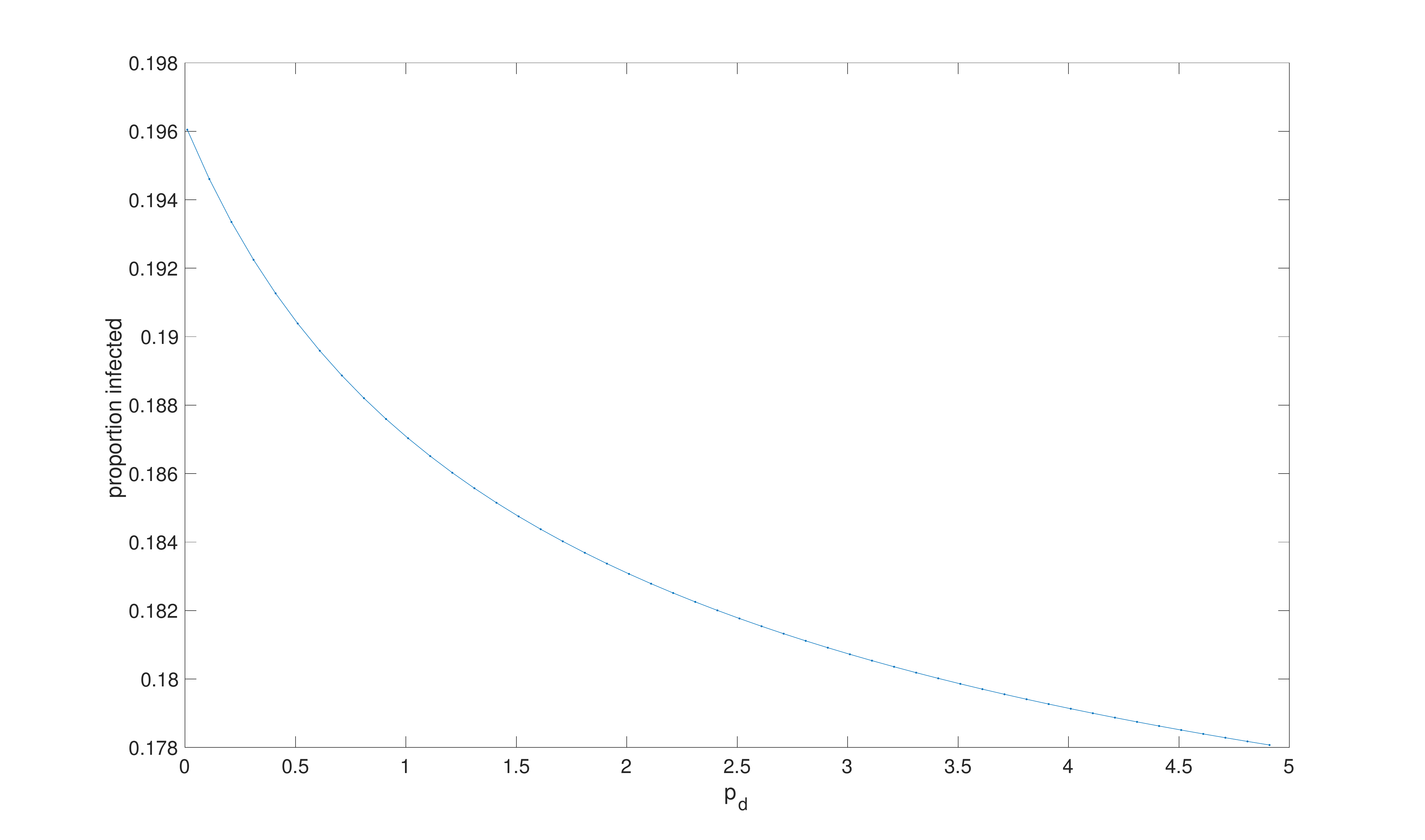}
    \caption{Graphing the proportion of infected humans against $p_d$.}
    \label{fig:Ih_vs_pd}
    \end{figure}
    
    Similarly, we vary $p_d$, $\mu$, and $p_h$ to examine the effect of these parameters on the proportion of infected humans, finding that while increasing the mutation and intermediate host-human contact rate increases this proportion, increasing $p_d$ lowers it (see Figures~\ref{fig:Ih_vs_pd},~\ref{fig:Ih_vs_mu}, and ~\ref{fig:Ih_vs_ph}), as a larger contact rate between wild and domestic animals leads to a larger proportion of animals infected with the non-human-transmissible strain and thus unable to pass the disease to humans. Further, for comparison, we vary $\beta_h$ from 0 to 5 using the same step length of 0.01. As shown in Figure ~\ref{fig:Ih_vs_bh}, while increasing $\beta_h$ can effect the proportion of infected humans, even decreasing $\beta_h$ to 0 still leads to an endemic equilibrium, with $I_h^*>0$.
    
    \begin{figure}[!ht]
    \centering
    \includegraphics[scale=0.25]{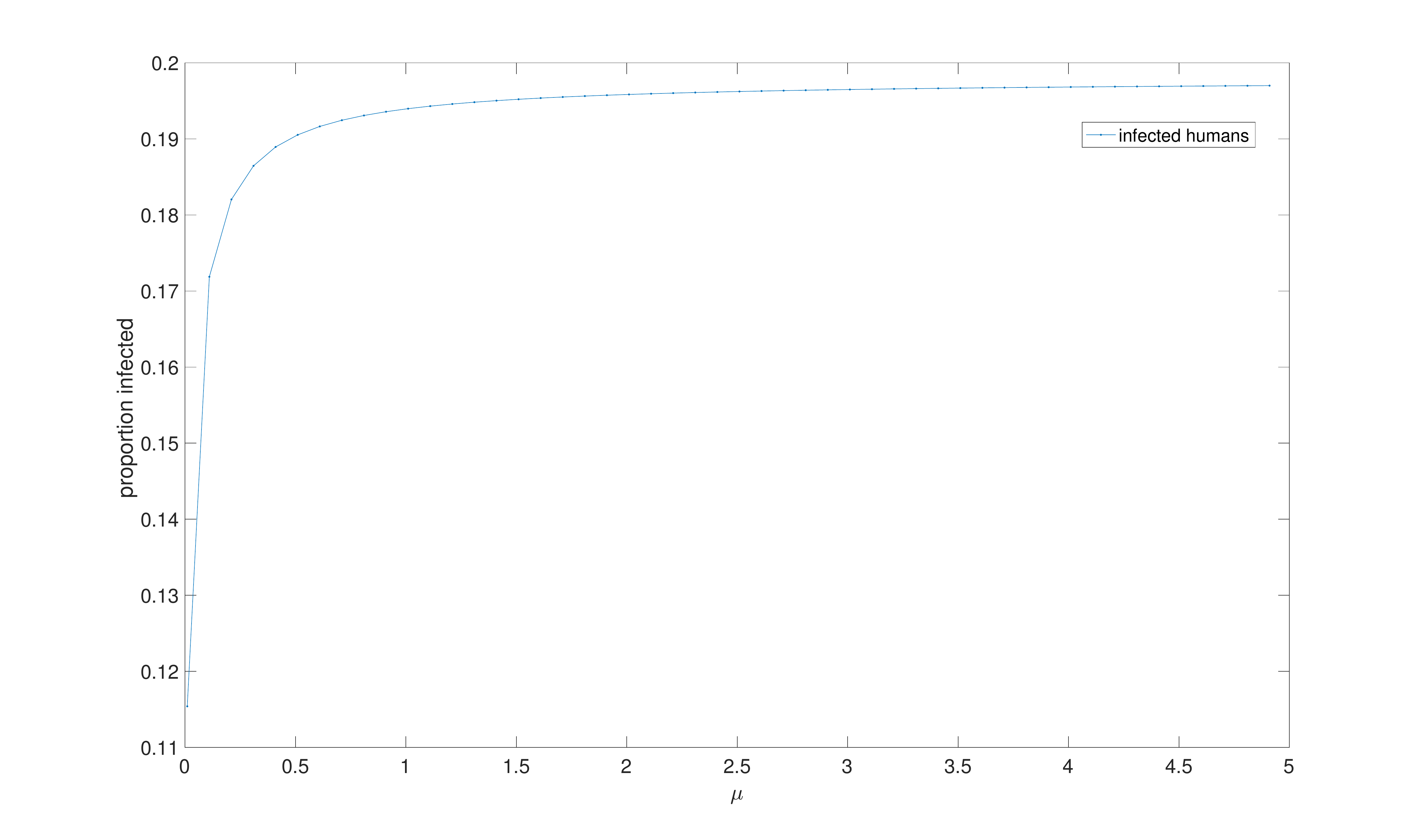}
    \caption{Graphing the proportion of infected humans against $\mu$, the rate of mutation from the wildtype strain to the human-transmissible strain.}
    \label{fig:Ih_vs_mu}
    \end{figure}
    
    \begin{figure}%[!ht]
    \centering
    \includegraphics[scale=0.25]{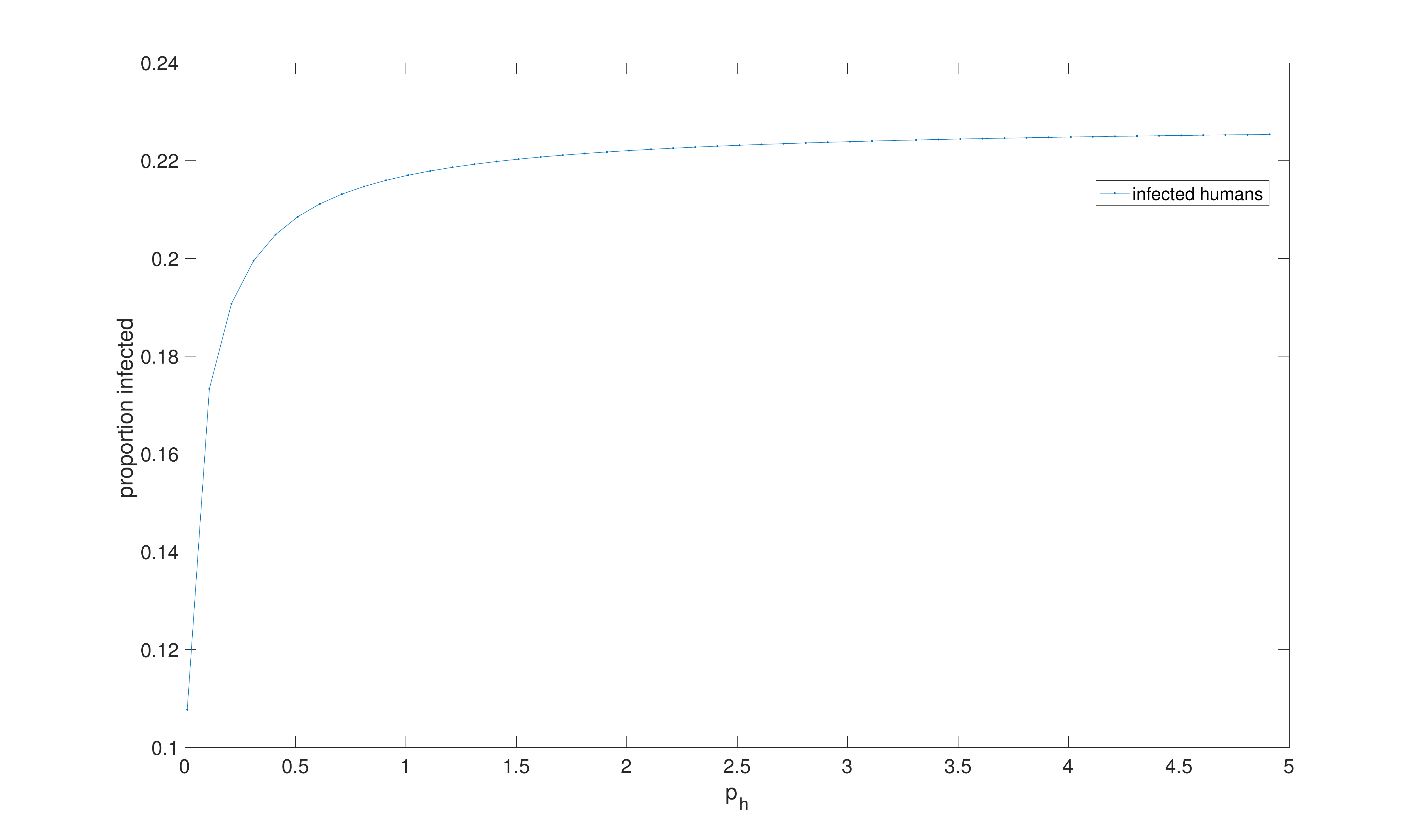}
    \caption{Graphing the proportion of infected humans against $p_h$, the contact rate (spillover rate) between domestic animals and humans.}
    \label{fig:Ih_vs_ph}
    \end{figure}
    
    \begin{figure}%[!ht]
    \centering
    \includegraphics[scale=0.25]{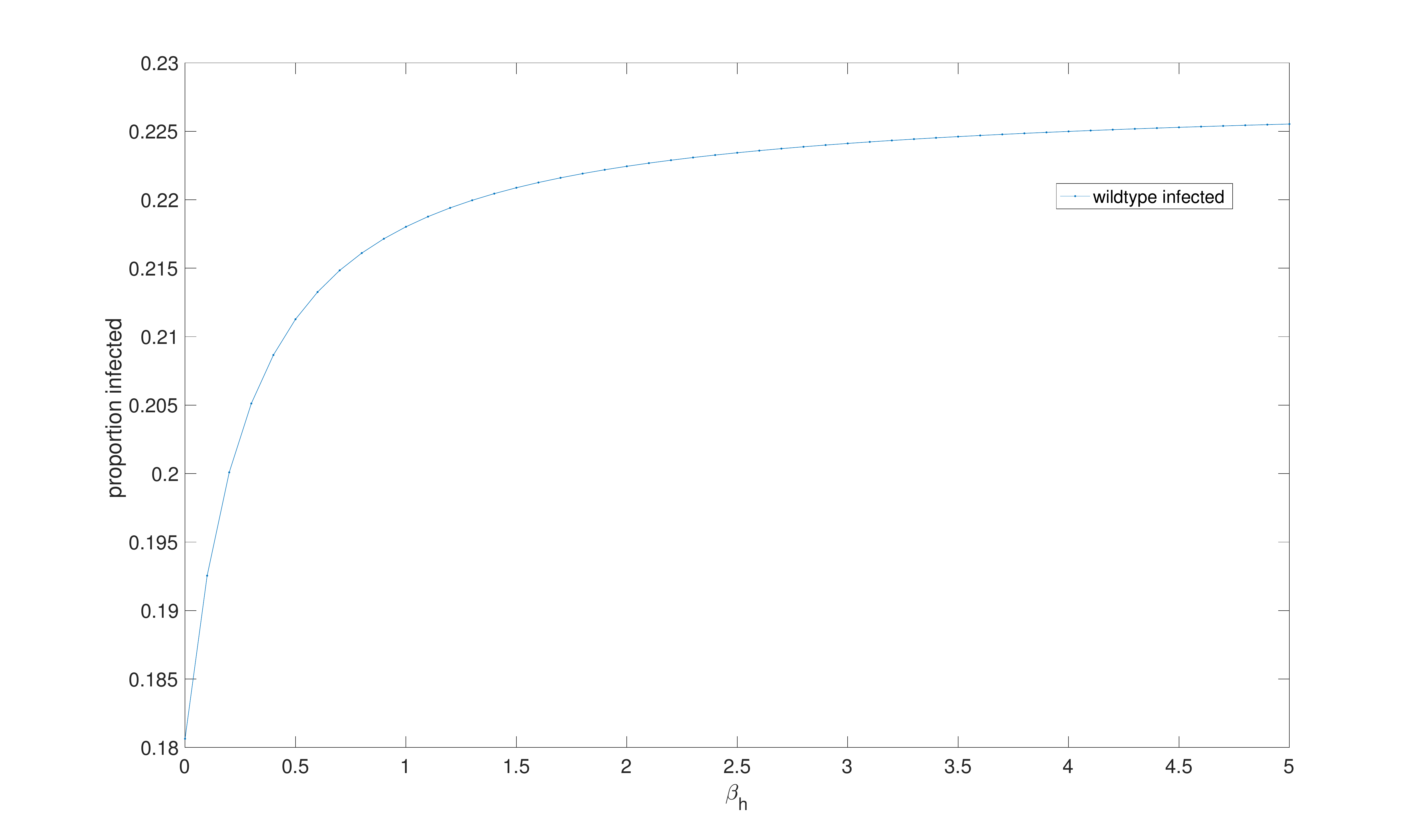}
    \caption{Graphing the proportion of infected humans against $\beta_h$, the transmission rate among humans. Here, setting $\beta_h$ to 0 still gives rise to an endemic equilibrium of infected humans.}
    \label{fig:Ih_vs_bh}
    \end{figure}
    
    \begin{figure}[!ht]
    \centering
    \includegraphics[scale=0.25]{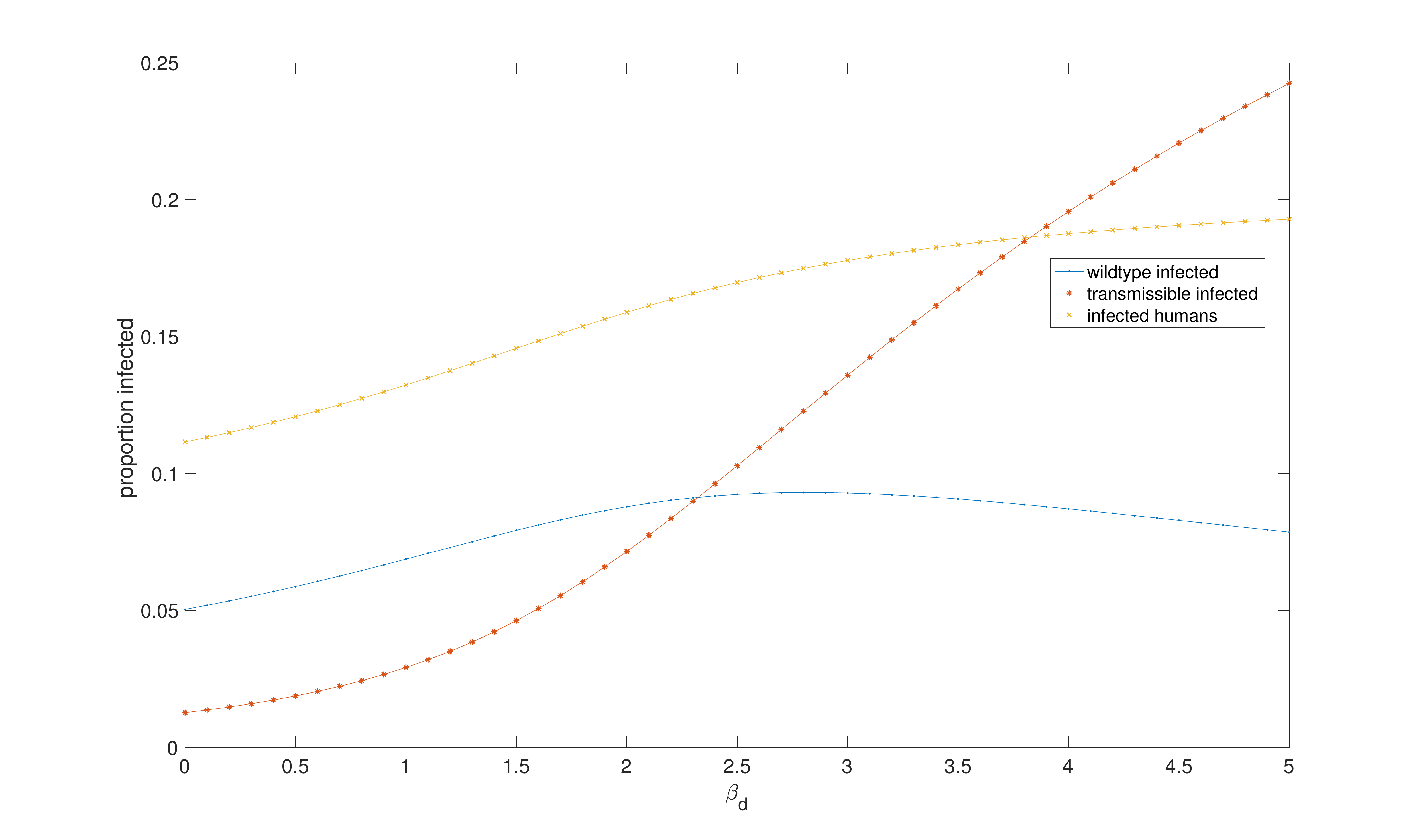}
    \caption{Graphing the proportion of infected humans and domestic animals against $\beta_d$, the transmission rate among domestic animals. Here, setting $\beta_d$ to 0 still gives rise to an endemic equilibrium of infected humans.}
    \label{fig:ITI_vs_bd}
    \end{figure}

    The importance of the interspecies transmission parameters is reflected in Figures \ref{fig:Ih_vs_bh} and \ref{fig:ITI_vs_bd}, which show that even when the transmission rates of the pathogen in humans or domestic animals is set to 0, the disease can reach an endemic equilibrium in humans. The effect of setting each parameter to in an otherwise endemic equilibrium, where the epidemic is expected to remain endemic in all three species, as in Figure \ref{fig:endemic}, is summarized in Table \ref{tab:paramcompare}.
    
    \begin{table}[!ht]
    \centering
    \begin{small}
    \begin{tabular}{| c | c | c | c|}
    \hline
    Parameter & Maximum \% of Infected Humans & Equilibrium \% of Infected Humans \\
    \hline
    -- & 46.33 & 19.04 \\
    $p_d$ & 0 & 0 \\
    $\mu$ & 0 & 0 \\
    $p_h$ & 0 & 0 \\
    $\beta_w$ & 49.77 & 19.62 \\
    $\beta_d$ & 18.81 & 11.16 \\
    $\beta_h$ & 36.13 & 18.06 \\
    \hline
    \end{tabular}
        \end{small}
    \caption{Comparing the effect of setting each transmission parameter to 0 in the endemic equilibrium of Figure \ref{fig:endemic}.}
    \label{tab:paramcompare}

    \end{table}
    
    These comparisons suggest that a lower number of animals infected with the transmissible strain has the potential to lower the proportion of infected humans, while even if the intraspecies parameters $\beta_d$ or $\beta_h$ are set to 0 the epidemic can spread to infect the human population. These results show that the interspecies transmission parameters are primary targets for intervention to lower the proportion of infected humans in this model.
    
\subsection{Summary}

To test the result from section 3 that changing $\mu$, $p_h$, and $p_d$ matters more to the eventual number of infected humans than changing $\beta_i$ or $\gamma_i$, the traditional parameters targeted in public health interventions, we varied the parameters $p_d, \mu, p_h, \beta_d,$ and $\beta_h$ while holding the other values constant at an endemic equilibrium condition. The results of these numerical simulations show that varying $p_d$ and $\mu$ can change the relative prevalence of domestic animals infected with the wildtype and human-transmissible strains, which in turn can change the proportion of infected humans. Further, only by setting one or more interspecies transmission parameters $\mu, p_d, p_h$ to 0 can the model avoid an endemic equilibrium in humans. In particular, the pathogen can persist in humans even if $\beta_h = 0$.

While varying traditional epidemic parameters such as $\beta_i$ and $\gamma_i$ can change the relative numbers of individuals in each compartment, section 3 shows that only $p_d$, $p_h$, and $\mu$ control the global behavior of a zoonotic epidemic, a result detailed by the simulations in this section. These results show that a zoonotic pathogen can establish itself in the human population as long as it is seeded with an initial infection in the wild compartment and $p_d, p_h$ and $\mu$ are nonzero, even if the human-transmissible strain is incapable of being transmitted between humans.

\section{Discussion} \label{ch:discussion}

The results of the mathematical analysis in section 3 suggest that we can categorize the parameters of the model into two types. The first type is intracompartment parameters: the transmission and recovery rates $\beta_i$ and $\gamma_i$, which describe interactions in a single species. The second is intercompartment parameters, which govern interactions between members of two species. $p_d$ and $p_h$, which indicate the spillover rate to domestic animals and humans, obviously fall into this category; $\mu$ quantifies the rate of a mutation arising in domestic animals that makes the pathogen transmissible among humans, and is thus also included. From Theorems 2 and 7, we see that it is only these second parameters, and the initial proportion of infected wild animals, that have the potential to alter the global dynamics of the three-species system to a disease-free equilibrium. The examples in section 4 crystallize the result that parameters of the second type are threshold values for the global progression of an epidemic: changing values in the first category only changes the relative proportions of each type of individual present at an equilibrium, not the stability of the equilibria, while changing the values of parameters in the second category can change the global behavior of the pathogen.

This complete simulation of an emerging zoonosis shows that even in cases where the disease dies out in the wild compartment and would fail without an external force of infection in the domestic one, it can establish an endemic equilibrium in humans. Further, this result holds even if $\beta_h = 0$, reflecting a pathogen in Stage 1 of the traditional categorization for zoonoses that would not be deemed a pandemic threat under that framework. While the high endemic prevalence seen in Figure \ref{tab:paramcompare} may be due to an overestimate of the pathogen's transmissibility in humans$-$the model assumes that all humans have an equally high exposure to poultry, while in reality agricultural workers are the group most at risk$-$these simulations suggest that the threat posed by zoonoses is more severe than previously assumed. Only by entirely suppressing at least one of the transmission parameters, an extraordinarily difficult feat, can public health officials prevent a pathogen with an intermediate host from establishing a presence in humans. This result indicates that even the slightest possibility of contact between species or selection for a pathogen more suited to humans raises $p_d, p_h$, or $\mu$ above 0 and thus can lead to an endemic infection in humans. While this endemic equilibrium or rates of transmission may be negligible in real populations, our results that the threat of an emerging zoonosis cannot be completely erased even with extraordinarily effective public health and medical interventions, confirming the focus on prioritizing zoonoses as mathematically sound and offering a warning for public health officials.

\subsection{Future Research}

The lack of large, publicly available data sets, especially regarding the prevalence of zoonotic infections in wild and domestic animals and the values for $p_d$, $p_h$, and $\mu$, limits our ability to refine any model
\citep{allen2012mathematical}. While some research attempts an explicit response to the lack of such by assessing expert approximations \citep{singh2018assessing}, this type of analysis cannot replace population-level data. Gathering such data is thus critical to future modeling efforts in domestic and wild animal populations (\citep{lloyd2015nine}, \citep{lloyd2009epidemic})$-$in particular, there is little data available for any infectious diseases in wild animals and interspecies contact rates$-$and should form a key component of future efforts.

This research introduces a model capable of replicating all stages of the emergence of a zoonosis with an intermediate host. To keep this work at a preliminary level and to maximize its use in more specialized contexts, we have not considered further modifications to the SIR prototype model such as loss of immunity (SIRS) or exposure time (SEIR), or possible variation patterns in the number of infected reservoir hosts, such as seasonal migration. Given adequate data, future research could add any of these modifications, as well as others not considered here, and can thus adapt this model to any specific emerging zoonosis. More specifically, future models should incorporate backwards transmission to wild animals, direct interactions between humans and wild reservoirs, as well as interactions between different pathogens in an intermediate host
\citep{lloyd2009epidemic}. Finally, since not all humans have the same level of exposure to a given intermediate host species, a more refined model could relax the assumption of mass action in the human compartment. The modifications discussed above have the potential to introduce more exciting dynamics, such as backward bifurcations or strange attractors in the solution space \citep{barrientos2017chaotic}, a type of behavior that could have implications for the policies governing zoonosis interventions. 

In particular, future models should investigate the effect of different transmission rates for the two strains circulating in the intermediate host, which will change the endemic equilibrium in domestic animals and thus humans. Here, we have abstracted the process of mutation to a binary question regarding human transmissibility, neglecting the distinction between the different possible ways for a pathogen to mutate and the different possible degrees of change. The mutation rate of a zoonotic disease can depend on social factors such as culling in the intermediate host population, vaccination of infected individuals, and biosecurity, as well as biological ones such as RNA mutation or interstrain competition \citep{goodwin2012interdisciplinary}, and future research should investigate whether those different processes have noticeable differences on the number of infected humans or the mathematical structure of the model. There is also a lack of investigation of disease dynamics in individual hosts, with little data investigating the effect of different expressions of pathogen genotypes or animal superspreaders (individuals who infect many more secondary cases than average) on transmissibility in humans \citep{lloyd2009epidemic}. As this effect is the one abstracted by our parameter $\mu$, delving deeper into individual-host pathogen dynamics such as cellular entry and replication \citep{allen2012mathematical} has the potential to improve our model. No emerging infected disease has been predicted before infecting humans \citep{morse2012prediction}, although progress is being made on identifying disease `hotspots' \citep{daszak2012anatomy}, and this inability reinforces the importance of studying the factors that lead to successful spillover and define transmission rates between species \citep{morse2012prediction}. 

This research suggests future avenues of exploration for both researchers and policymakers seeking to understand and control the spread of an emerging infectious zoonosis, and proves that interspecies connections are critical to controlling and understanding the effect an emerging zoonosis can have on human populations. We show that with nonzero transmission parameters and an initial population of infected wild animals, a pathogen can fail to achieve traditional markers of success, such as stage 3 transmissibility, and still maintain an endemic equilibrium in the human population. This is a concerning result for public health, but offers areas in which policy rather than medical interventions can be more effective in controlling disease.

\section{Conclusion} \label{ch:conclusion}

With the ability to study the emergence of a zoonosis with an intermediate host, first quantified by the model introduced here, scientists and policymakers alike have a more refined tool with which to study and confront one of the most well-recognized threats to global health: the emergence of a new pandemic into the human population. To our knowledge, this is the first model that accounts for the entire course$-$from infected wild animals, through mutation in an intermediate host, to an endemic equilibrium in humans$-$of the type of zoonotic pathogen the World Health Organization ranks in the highest tier of priorities for research and development, and so provides a significant step forward in its study.

We establish that the model has one unique disease-free equilibrium and one endemic equilibrium, and that the stability of these points depends on $p_d, p_h$, and $\mu$, the contact rates between species and the pathogen's rate of mutation. Accurately identifying and describing the dynamics of a pathogen circulating in wild and domestic animals provides an invaluable opportunity to avoid risk to humans \citep{morse2012prediction}, and can be used to guide public health interventions for emerging zoonotic diseases.

That the interspecies transmission parameters are the only threshold conditions for this model suggests that the problem of controlling the spread of a zoonotic epidemic has less to do with intracompartment controls than with intercompartment ones: rather than efforts to control the transmission or recovery rates in one species, it is a more effective intervention to control $p_d, p_h,$ or $\mu$ through better biosecurity or population control. This finding provides a blueprint for public health interventions in zoonoses, as well as a warning for officials hoping to prevent the spread of wildlife diseases to humans. The interspecies parameters$-p_d, p_h$, and $\mu-$may be more susceptible to policy changes than the intraspecies parameters $\beta_i$ and $\gamma_i$, at least when the domestic intermediate host is a livestock or pet species entirely under human control \citep{cunningham2017one}. Even before a zoonotic epidemic is detected in other species, restructuring agricultural systems and controlling livestock movements offer public health policymakers avenues to mitigate the effects of such a pathogen. For example, by preventing disease circulation on farms, we can prevent pathogens such as avian influenza from becoming persistent human health risks (\cite{karesh2012ecology}, \cite{morse2012prediction}). Since accurate models can assist in appropriately allocating surveillance resources \citep{lloyd2015nine}, these parameters can thus guide health officials in their response to and prevention of emerging zoonoses, policy changes which are essential in mitigating the risks of such diseases \citep{cunningham2017one}.

Our results primarily offer a warning to public health officials: without drastic interventions to drive interspecies interactions or pathogen mutation rates to 0, which may be biologically impossible, zoonoses with the capacity to mutate in a human-adjacent intermediate host can spread to humans even if they are not viable in a human population alone. More fundamentally to the field of mathematical epidemiology, this result confirms previously held beliefs$-$unquantified until now$-$about the philosophical importance of zoonoses to humanity. It is a pillar of the movement variously called ``global'', ``one'', or ``planetary health'' that human populations cannot isolate themselves from changes that affect other species with interventions targeting only humans. By mathematically linking the progress of a zoonotic epidemic to parameters governing interactions between species, this model shows that the framework of an interconnected human and natural world that implicitly underlies much of the analysis in this field in the last twenty years agrees with the mathematics of infectious disease, quantifying and confirming a widespread belief in global health.

\section*{Acknowledgments}

We are grateful for generous support from the Neukom Institute at Dartmouth, the Neukom CompX Faculty Grant, Walter \& Constance Burke Research Initiation Award, and NIH COBRE Award. F.F. acknowledges a Junior Faculty Fellowship funded by the Dartmouth's Dean of the Faculty for this work.

\end{document}